\documentclass[a4paper,onecolumn,titlepage,superscriptaddress,11pt,accepted=2018-05-12]{quantumarticle}
\pdfoutput=1
\usepackage{hyperref}
\usepackage[numbers,sort&compress]{natbib}

\usepackage{graphicx}
\usepackage{makeidx}              
\usepackage{array}

\usepackage{multirow}

\usepackage{amsthm,amssymb,amsfonts,bbm}
\usepackage{amsmath}
\usepackage[usenames,dvipsnames,svgnames,table]{xcolor}
\usepackage{hhline}
\usepackage{theoremref}
\usepackage{tabularx}
\usepackage{bm}
\newcommand{\barred}[1]{\mathbbm{#1}}

\newcommand{\ints}{\barred{N}}

\newcommand{\LL}{{\cal L}}
\renewcommand{\P}{{\cal P}}
\newcommand{\C}{{\cal C}}
\newcommand{\Q}{{\cal Q}}
\newcommand{\NS}{{\cal C}}

\newcommand{\supp}{{\rm supp}} 

\newcommand{\ZO}{\{0,1\}}

\newcommand{\R}{\mathrm{R}}

\newcommand{\DISJ}{\mathrm{DISJ}}
\newcommand{\TRIBES}{\mathrm{TRIBES}}
\newcommand{\GHD}{\mathrm{GHD}}
\newcommand{\ORT}{\mathrm{ORT}}
\newcommand{\VSP}{\mathrm{VSP}}

\renewcommand{\det}{\mathrm{det}}

\newcommand{\abs}[1]{\mid\! #1\! \mid}

\newcommand{\Real}{\mathbb R}
\newcommand{\eps}{\varepsilon}

\newcommand{\A}{\mathcal{A}}
\newcommand{\B}{\mathcal{B}}

\newcommand{\X}{\mathcal{X}}
\newcommand{\Y}{\mathcal{Y}}

\newcommand{\tr}{{\rm tr}}

\newcommand{\eff}{\mathrm{\bf eff}}

\newcommand{\bp}{\mathbf{p}}
\newcommand{\bl}{\bm{\ell}}
\newcommand{\bq}{\mathbf{q}}
\newcommand{\bbm}{\mathbf{m}}
\newcommand{\bdelta}{\boldsymbol{\Delta}}

\newcommand{\ket}[1]{| #1 \rangle}
\newcommand{\bra}[1]{\langle #1|}

\newtheorem{definition}{Definition}
\newtheorem{theorem}{Theorem}
\newtheorem{lemma}{Lemma}

\newtheorem{observation}{Observation}
\newtheorem{corollary}{Corollary}
\newtheorem{proposition}{Proposition}
\newtheorem{remark}{Remark}
\newtheorem{fact}{Fact}

\graphicspath{{../Pictures/}}

\title{Robust Bell inequalities from communication complexity}
\author{Sophie Laplante}
\affiliation{IRIF, Universit\'e Paris-Diderot, Paris, France}
\email{laplante@irif.fr}
\author{Mathieu Lauri{\`e}re}
\affiliation{ORFE, Princeton University, Princeton, NJ 08544, USA}
\email{lauriere@princeton.edu}
\author{Alexandre Nolin}
\affiliation{IRIF, Universit\'e Paris-Diderot, Paris, France}
\email{nolin@irif.fr}
\author{J{\'e}r{\'e}mie Roland}
\affiliation{Universit\'e Libre de Bruxelles, Brussels, Belgium}
\email{jroland@ulb.ac.be}
\author{Gabriel Senno}
\affiliation{ICFO-Institut de Ciencies Fotoniques, The Barcelona Institute of Science and Technology, 08860 Castelldefels (Barcelona), Spain}
\email{gabriel.senno@icfo.eu}
\thanks{\vspace{1cm}\\A first version of this work appeared in \cite{laplante_et_al:LIPIcs:2016:6686}.}

\begin{document}

\maketitle

\begin{abstract}
The question of how large Bell inequality violations can be, for quantum distributions, has been the object of much work in the past several years. We say that a Bell inequality is normalized if its absolute value does not exceed 1 for any classical (\textit{i.e.} local) distribution. Upper and (almost) tight lower bounds have been given for the quantum violation of these Bell inequalities in terms of number of outputs of the distribution, number of inputs, and the dimension of the shared quantum states. In this work, we revisit normalized Bell inequalities together with another family: inefficiency-resistant Bell inequalities. To be inefficiency-resistant, the Bell value must not exceed 1 for any local distribution, including those that can abort. This makes the Bell inequality resistant to the detection loophole, while a normalized Bell inequality is resistant to general local noise. Both these families of Bell inequalities are closely related to communication complexity lower bounds.
We show how to derive large violations from any gap between classical and quantum communication complexity, provided the lower bound on classical communication is proven using these lower bound techniques.
This leads to inefficiency-resistant violations that can be exponential in the size of the inputs. Finally, we study resistance to noise and inefficiency for these Bell inequalities.
\end{abstract}


\section{Introduction}

The question of achieving large Bell violations has been studied since Bell's seminal paper
in 1964~\cite{Bell64}. In one line of investigation, proposals have been made to exhibit families of distributions which admit unbounded violations \cite{Mermin90,LPZ+04,NLP06,PWP+08}. In another, various measures of nonlocality have been studied, such as the amount of communication necessary and sufficient to simulate quantum distributions classically \cite{Mau92,BCT99,Ste00,TB03,Pir03,DKLR11,CM14,BG11,BC17}, or the resistance to detection inefficiencies  and noise. More recently, focus has turned to giving upper and lower bounds on violations achievable, in terms of various parameters: number of players, number of inputs, number of outputs, dimension of the quantum state, and amount of entanglement \cite{DKLR11,JPP+10, JP11}.

Up until quite recently, violations were studied in the case of specific distributions (measuring Bell states), or families of distributions. 
Buhrman \textit{et al.}~\cite{BRSdW12} gave a construction that could be applied to several problems which had efficient quantum protocols (in terms of communication) and for which one could show a trade-off between communication and error in the classical setting. This still required an {\em ad hoc} analysis of communication problems. Recently Buhrman \textit{et al.}~\cite{BCG+16} proposed the first general construction of quantum states along with Bell inequalities from any communication problem. The quantum states violate the Bell inequalities when there is a sufficiently large gap between quantum and classical communication complexity (a super-quadratic gap is necessary, unless a quantum protocol without local memory exists).

Table~\ref{table:summaryKnown} summarizes the best known upper and lower bounds on quantum violations achievable with normalized Bell inequalities.

\begin{table}[h]
\begin{center}
{\renewcommand{\arraystretch}{2}
\begin{tabularx}{\linewidth}{|l|>{\centering\arraybackslash}X|>{\centering\arraybackslash}X|>{\centering\arraybackslash}X|}
\hline
Parameter&Upper bound& \parbox{2.9cm}{\centering Ad hoc lower bounds}& \parbox{2.9cm}{\centering Best possible lower bound from~\cite{BCG+16}}\\[2pt]
\hline
\hline
Number of inputs $N$	& $2^{c} \leq N$~\cite{LS09,DKLR11,JPP+10} &$\frac{\sqrt{N}}{\log(N)}$~\cite{JP11}& $\frac{\sqrt{c}}{q}\leq \log(N)$\\
\hline
Number of outputs~$K$	&$O(K)$~\cite{JP11}&$\Omega\left(\frac{K}{(\log(K))^2}\right)$~\cite{BRSdW12}&$\leq \log(K)$\\
\hline
Dimension $d$	&$O(d)$~\cite{JPP+10} &$\Omega\left(\frac{d}{(\log(d))^2}\right)$~\cite{BRSdW12} &$\leq \log\log(d)$\\
\hline
\end{tabularx}
}
\end{center}
\caption{Bounds on quantum violations of bipartite normalized Bell inequalities, in terms of the dimension $d$ of the local Hilbert space, the
number of settings (or inputs) $N$ and the number of outcomes (or outputs) $K$ per party. In the fourth column, we compare ad hoc results to the recent constructions of~\cite{BCG+16} (Theorem~\ref{thm:Buhrmanetal}) which gives a lower bound of~$\frac{\sqrt{c}}{q}$, where $c$ (resp. $q$) stands for the classical (resp. quantum) communication complexity of simulating a distribution. We give upper bounds on their construction in terms of the parameters $d,N,K$.}
\label{table:summaryKnown}
\end{table}

\subsection{Our results}

We revisit the question of achieving large Bell violations by exploiting known
connections with communication complexity.
Strong lower bounds in communication complexity, equivalent to the partition bound,
amount to finding {\em inefficiency-resistant Bell inequalities}~\cite{LLR12}.
These are Bell functionals that
are bounded above by 1 on all local distributions {\em that can abort}, i.e. local distributions with an additional abort outcome $\bot$ for each party.
In an experimental Bell test, such a constraint would provide resistance to the detection loophole, by associating an event where one of the detectors does not click to an abort outcome.

First, we study the resistance of normalized Bell inequalities to inefficiency.
We show that, up to a constant factor in the value of the violation, any
normalized Bell inequality can be made resistant to inefficiency while
maintaining the normalization property ({\bf Theorem~\ref{thm:BellLLbot}}).

Second, we show how to derive large Bell violations from any communication
problem for which the partition bound is bounded below and the quantum
communication complexity is bounded above.
The problems studied in
communication complexity are far beyond the quantum set, but we show how to easily derive
a quantum distribution from a quantum protocol. The Bell value we obtain is
$2^{c-2q}$, where $c$ is the partition lower bound on the classical communication complexity of the problem considered,
and $q$ is an upper bound on its quantum communication complexity~({\bf Theorem~\ref{thm:existsBqeff} and Corollary~\ref{cor:existsBq}}).
The quantum distribution has one extra output per player compared to the original distribution and
uses the same amount of entanglement as the quantum protocol plus as many EPR pairs
as needed to teleport the quantum communication in the protocol.
Next, we show that the magnitude of the violations is unaltered in the presence of white noise, i.e. when the ideal quantum state giving rise to the correlations gets added a maximally mixed stated with some probability $\delta$ ({\bf Theorem~\ref{thm:uniform}}).

Finally, we provide tools to build Bell inequalities
from communication lower bounds in the literature.
Lower bounds used in practice to separate classical from quantum communication complexity
are usually achieved using corruption bounds and its variants. In {\bf Theorem~\ref{thm:boundDistribBellGeneral}},
we give an explicit construction which translates these bounds into a suitable Bell
functional.
Table~\ref{table:summaryResults} summarizes the new results or the improvements that we obtain in this work.

\begin{table}[h]
\begin{center}
{\renewcommand{\arraystretch}{2}
\begin{tabularx}{\linewidth}{| l | >{\centering\arraybackslash}X | >{\centering\arraybackslash}X |}
  \hline
  Problem & Normalized Bell violations \cite{BCG+16} &   \parbox{5cm}{\centering Inefficiency-resistant Bell violations (this work)} \\[2pt]
  \hline
  \hline
  VSP \hfill\cite{Raz99,KR11}  &       \parbox{4.5cm}{\centering$\Omega\left({\sqrt[6]{n}}/{\sqrt{\log n}}\right)$ \\ $ d~=~2^{\Theta(n\log n)}, K=2^{\Theta(n)}$}
                &       \parbox{5cm}{\centering$2^{\Omega(\sqrt[3]{n})-O(\log n)}$ \\ $ d~=~2^{O(\log n)}, K=3$}                \\
  \hline
  DISJ \hfill \cite{Razborov92,Razborov03,AA05}         &       N/A 
                &       \parbox{5cm}{\centering$2^{\Omega(n)-O(\sqrt n)}$ \\ $d~=~2^{O(\sqrt{n})}, K=3$}    \\
  \hline
  TRIBES \hfill \cite{JKS03,BCW98}       &       N/A 
                &       \parbox{5cm}{\centering$2^{\Omega(n)-O(\sqrt{n}\log^2 n)}$ \\ $ d~=~2^{O(\sqrt{n}\log^2 n)},K=3$}             \\
  \hline
  ORT \hfill\cite{She12,BCW98}          &       N/A 
                &       \parbox{5cm}{\centering$2^{\Omega(n)-O(\sqrt{n}\log n)}$ \\ $ d~=~2^{O(\sqrt{n}\log n)}, K=3$}         \\
  \hline
\end{tabularx}
}
\end{center}
\caption{Comparison of the Bell violations obtained by the general construction of Buhrman \textit{et al.}~\cite{BCG+16} for normalized Bell violations (second column) and this work, for inefficiency-resistant Bell violations
(see Propositions~\ref{cor:VSP},~\ref{cor:DisjBellIneq},~\ref{coro:TribesBellIneq}, and~\ref{coro:ORTBellIneq}).
The parameter $n$ is the size of the input (typically, $N=2^n$.) See Section \ref{subsec:examples} for the definitions of the communication problems appearing in this table. The construction of Buhrman \textit{et al.}~only yields a violation when the gap between classical and quantum complexities is more than quadratic. In the case where the gap is too small to prove a violation, we indicate this with ``N/A''.}
\label{table:summaryResults}
\end{table}

\subsection{Related work}

The study of the maximum violation of Bell inequalities began with Tsirelson~\cite{Tsirelson87}, who showed that for two-outcome correlation Bell inequalities, the maximum violation is bounded above by Grothendieck's constant. Tsirelson also raised the question of whether one can have unbounded violations of Bell inequalities. More precisely, he asked whether there exist families of Bell inequalities for which the amount of the violation grows arbitrarily large.

The first answer to this question came from Mermin~\cite{Mermin90}, who gave a family of Bell inequalities for which a violation exponential in the number of parties is achieved. In the years that followed, several new constructions appeared for number of parties and number of inputs~\cite{Ard92,Mas02,LPZ+04,NLP06,PWP+08}.

The study of upper bounds on violations of normalized Bell inequalities resumed in~\cite{DKLR11}, where an upper bound of $O(K^2)$ (with $K$ the number of outputs per player) and of $2^{c}\leq N$ (with $c$ the communication complexity and $N$ the number of inputs per player) were proven. In \cite{JPP+10} the authors proved a bound of $O(d)$ in terms of the dimension $d$ of the local Hilbert space, and in~\cite{JP11}, the bound in terms of the number of outputs was improved to~$O(K)$. 
In~\cite{JP11}, Bell inequalities are constructed for which a near optimal, but probabilistic, violation of order $\Omega(\sqrt{m}/\log{m})$, with $N=K=d=m$, is proven. In~\cite{BRSdW12}, the same violation, although requiring $N=2^m$ inputs, is achieved for a family of Bell inequalities and quantum distributions built using the quantum advantage in one-way communication complexity for the Hidden Matching problem (with $K=d=m$).
In the same paper, a violation of order $\Omega(m/(\log m)^2)$, with $K=d=m$ and $N=2^m/m$ is achieved with the Khot-Vishnoi game; later, Junge et al. proved in \cite{junge2016reducing} that $N$ can be reduced to $O(m^8)$ . Recently, an asymmetric version of this same game was introduced to allow one of the parties to only make dichotomic measurements, with a smaller  (although almost optimal for this scenario) violation $\Omega(\sqrt{m}/(\log m)^2)$~\cite{PY15}.

For \emph{inefficiency-resistant} Bell inequalities, the bounds in~\cite{JP11} do not apply. In fact, Laplante \textit{et al.}~proved in~\cite{LLR12} a violation exponential in the dimension and the number of outputs for this type of Bell functionals, achieved by a quantum distribution built, as in~\cite{BRSdW12}, from the Hidden Matching communication complexity problem.

The connection exhibited in~\cite{BRSdW12} between Bell violations and communication complexity is generalized by Buhrman \textit{et al.}~in~\cite{BCG+16} where a fully general construction is given to go from a quantum communication protocol for a function~$f$ to a Bell inequality and a quantum distribution violating it when the gap between classical and quantum communication complexity is larger than quadratic. The downside to this construction is that the quantum distribution has a
double exponential (in the communication) number of outputs and the protocol to implement it uses an additional double exponential amount of entanglement. Also, this result does not apply for quantum advantages in a zero-error setting.

\section{Preliminaries}


\subsection{Quantum nonlocality}

Local, quantum, and nonsignaling distributions have been
widely studied in quantum information theory since the seminal paper of Bell~\cite{Bell64}.
In an experimental setting, two players share an entangled state
and each player is given a measurement to perform.  The outcomes
of the measurements are predicted by quantum mechanics and
follow some probability distribution $p(a,b|x,y)$, where $a$ is the
outcome of Alice's measurement $x$, and $b$ is the outcome of Bob's
measurement $y$.

We consider bipartite distribution families of the form
$\bp=(p(\cdot,\cdot|x,y))_{(x,y)\in \X\times\Y}$
with inputs $(x,y)\in \X\times\Y$
determining a probability distribution $p(\cdot,\cdot|x,y)$
over the outcomes $(a,b)\in \A\times\B$,
with the usual positivity and normalization constraints.
The set of probability distribution families is denoted by $\P$.  For simplicity, we call simply {``distributions"} such probability distribution families.
The expression {``Alice's marginal"} refers to her marginal output distribution, that is $\sum_b p(\cdot,b|x,y)$ (and similarly for Bob).

The {\em local
deterministic distributions}, denoted $\LL_{\det}$, are the ones where
Alice outputs according to a deterministic strategy, i.e., a (deterministic)
function of $x$, and Bob independently outputs as a function of $y$,
without communicating.
The {\em local distributions} $\LL$ are obtained by taking distributions over the local
deterministic strategies. Operationally, this corresponds to protocols with shared randomness
and no communication. Geometrically, $\LL$ is the convex hull of $\LL_{\det}$.

A {\em Bell test}~\cite{Bell64} consists of estimating all the
probabilities $p(a,b|x,y)$ and computing a {\em Bell functional},
or linear function, on these values.  The Bell functional $B$
is chosen together with a threshold $\tau$  so that any local classical
distribution $\bl$ satisfies the {\em Bell inequality}
$B(\bl)\leq \tau$, but the chosen
distribution $\bp$ exhibits a {\em Bell violation}: $B(\bp) > \tau$.
By normalizing $B$, we can assume without loss of generality
that $\bl$ satisfies $B(\bl)\leq 1$ for any $\bl\in \LL$, and $B(\bp)>1$.

In this paper, we will also consider strategies that are allowed
to abort the protocol with some probability.  When they abort, they output the
symbol $\bot$  ($\bot$ denotes a new symbol which is not in $\A\cup \B$).
We will use the notation $\LL_{\det}^\bot$  and $\LL^\bot$ to
denote local strategies that can abort, where $\bot$ is added to the possible outputs
for both players.
When~$\ell \in \LL_{\det}^\bot$ or $\LL^\bot$, $\ell(a,b|x,y)$ is {\em not} conditioned on
$a,b\neq \bot$ since $\bot$ is a valid output for such distributions.

The {\em quantum distributions}, denoted $\Q$, are the ones that result from
applying measurements $x,y$ to their part of a shared bipartite quantum state.  Each player outputs his or her  measurement
outcome ($a$ for Alice and $b$ for Bob).  In communication complexity terms, these are
zero-communication protocols with shared entanglement.
If the players are allowed to abort, then the corresponding
set of distributions is denoted $\Q^\bot$.

Boolean (and other) functions can be cast as sampling problems.
Consider a boolean function~$f:\X\times \Y \rightarrow \{0,1\}$
(nonboolean functions and relations can
be handled similarly).
First, we split the output so that if $f(x,y)=0$, Alice
and Bob are required to output the same bit, and
if $f(x,y)=1$, they output different bits.
Let us further require Alice's marginal distribution to be
uniform, likewise for Bob, so that the distribution is well defined.
Call the resulting distribution
$\bp_f$, that is, for any $a,b\in\{0,1\}$ and $(x,y)\in\X\times\Y$, we have
$p_f(a,b|x,y)= 1/2$ if $a\oplus b=f(x,y)$, and
$p_f(a,b|x,y)= 0$ otherwise,
$\oplus$ being the $1$-bit XOR.

If $\bp_f$ were local, $f$
could be computed with one bit of communication
using shared randomness: Alice sends her output to Bob,
and Bob XORs it with his output.
If $\bp_f$ were quantum, there would be a 1-bit
protocol with shared entanglement for $f$.  Interesting distributions
have nontrivial communication complexity and, therefore, they usually
lie well beyond these sets.

Finally, a distribution is {\em nonsignaling} if for each player, its marginal output
distributions, given by $p_A(a|x,y)=\sum_b p(a,b|x,y)$, for Alice, and
$p_B(b|x,y)=\sum_a p(a,b|x,y)$, for Bob,  do not depend on the other player's input.
When this is the case, we write the marginals as $p_A(a|x)$ and $p_B(b|y)$. Operationally,
this means that each player cannot influence the statistics of what the other player
observes with his own choice of input. We note with $\NS$ the set of nonsignaling distributions, also referred to as the {\em causal} set, and we note $\NS^\bot$ when we allow aborting.
The well-known inclusion relations between these sets are $\LL \subset \Q \subset \NS\subset \P$.

For any Boolean function $f$, the distribution $\bp_f$ is
nonsignaling since the marginals are uniform.
A fundamental question of quantum mechanics
has been to establish experimentally whether nature is truly \emph{nonlocal}, as
predicted by quantum mechanics, or whether
there is a purely classical (i.e., \emph{local}) explanation to the phenomena that
have been predicted by quantum theory and observed in the lab.

\subsection{Measures of nonlocality}

We have described nonlocality as a yes/no property, but some distributions are somehow
more nonlocal than others. To have a robust measure of nonlocality, it should
fulfill some common sense properties: for a fixed distribution, the measure should
be bounded; it should also be convex, since sampling from the convex combination of two distributions
can be done by first picking randomly one of the two distributions using shared
randomness, and then sampling from that distribution. We also expect such a measure of nonlocality
to have various equivalent formulations. Several measures have been proposed and studied:
resistance to noise \cite{KGZ+00,ADG02,PWP+08,JPPG+10}, resistance to inefficiency \cite{Mas02,MPR+02,LLR12},
amount of communication necessary to reproduce them \cite{Mau92,BCT99,Ste00,TB03,Pir03,DKLR11},
information-theoretic measures \cite{BCSS11,GWLA+12,FWW09}, etc.

In the form studied in this paper, normalized Bell inequalities were first studied
in~\cite{DKLR11}, where they appeared as the dual of the linear program for a
well-studied lower bound on communication complexity, known as the nuclear  norm~$\nu$~\cite{LS09}
(the definition is given in Section~\ref{sec:cc-lb}). There are many
equivalent formulations of this
bound. For distributions arising from boolean functions, it has the mathematical properties
of a norm, and it is related to winning probabilities of XOR games. It can also be viewed as a gauge,
that is, a quantity measuring by how much the local set must be expanded in order to contain the distribution considered. For more
general nonsignaling distributions, besides having a geometrical interpretation in terms of affine combinations of local distributions, it has also been shown to be equivalent to the amount of local noise that can
be tolerated before the distribution becomes local~\cite{JPP+10}.

A subsequent paper~\cite{LLR12}  studied equivalent formulations of
the partition bound, one of the strongest lower bounds
in communication complexity~\cite{JK10}. This bound also also has several formulations: the primal formulation can be viewed as
resistance to detector inefficiency, and the dual formulation is given in terms
of inefficiency-resistant Bell inequality violations.

In this paper, we show how to deduce large violations on quantum distributions from large violations on nonsignaling distributions, provided there are efficient quantum
communication protocols for the latter.

\subsection{Communication complexity and lower bounds}
\label{sec:cc-lb}

In classical communication complexity (introduced by~\cite{yao1979classicalCC}), two players each have a share of the
input, and wish to compute a function on the full input. Communication complexity measures the number of bits they need to exchange to solve this problem in the worst case, over all inputs of a given size $n$.
In this paper we consider a generalization of this model, where instead of computing a function, they each produce an output, say~$a$ and~$b$, which should follow, for each $(x,y)$, some prescribed distribution $p(a,b|x,y)$ (which depends on their
inputs $x,y$). We assume that the order in which the players speak does not depend on the inputs. This is without loss of generality at a cost of a factor of 2 in the communication.

We use the following notation for communication complexity of distributions.
$R_\epsilon(\bp)$ is the
minimum number of bits exchanged 
in the worst case between players having access to shared randomness in order to output with
the distribution $\bp$
up to $\epsilon $ in total variation distance for all $x,y$.
We call total variation distance
between distributions the distance denoted by $|.|_1$, and defined as
$|\bp-\bp'|_1 = \max_{x,y}\sum_{a,b} | p(a,b|x,y) - p'(a,b|x,y)|$.
We use $Q_\epsilon$ to denote quantum
communication complexity (see~\cite{yao1993quantumCC}) and $Q_\epsilon^*$ when, in addition to the communication of quantum states, the players are allowed to share entanglement.

To give upper bounds on communication complexity it suffices to give
a protocol and analyze its complexity. Proving lower bounds is often
a more difficult task, and many techniques have been developed to achieve
this. The methods we describe here are complexity measures which can be
applied to any function. To prove a lower bound on communication,
it suffices to give a lower bound on one of these complexity measures,
which are bounded above by communication complexity for any function.
We describe here most of the complexity measures relevant to this work.

The nuclear  norm $\nu$,
given here in its dual formulation and extended to nonsignaling
distributions,  is expressed by
the following linear program~\cite{LS09, DKLR11}.
(There is a quantum analogue, $\gamma_2$, which is not needed in this work.
We refer the interested reader to the definition for distributions in~\cite{DKLR11}).
\begin{definition}[\cite{LS09,DKLR11}]
\label{def:nu}
The nuclear norm $\nu$ of a nonsignaling distribution $\bp\in\C$ is given by
\begin{align*}
 \nu(\bp)=&\max_{B} &   B(\bp)\\
&\textrm{subject to } &
\abs{B(\bl)}\leq 1 & \quad \forall \bl\in\LL_{det}.
\end{align*}
With error $\epsilon$, $\nu_\epsilon(\bp) = \min_{\bp'\in\NS:|\bp'-\bp|_1\leq\epsilon} \nu(\bp')$.
%
We call any Bell functional that satisfies the constraint in the above linear program {\em normalized Bell functional}.
\end{definition}


In this definition and in the rest of the paper, unless otherwise specified (in particular in Lemma~\ref{lem:removeabortinB}), $B$ ranges over vectors of
real coefficients $B_{a,b,x,y}$ and $B(\bp)$ denotes $\sum_{a,b,x,y} B_{a,b,x,y}p(a,b|x,y)$,
where $a,b$ ranges over the nonaborting outputs and $x,y$ ranges over the inputs. So even when $B$ and $\bp$ have coefficients on the abort events, we do not count them.
 Table~\ref{table:summaryKnown} summarizes the known upper and lower bounds on $\nu$ for
various parameters.
The (log of the) nuclear norm is a lower bound on classical communication complexity.
\begin{proposition}[\cite{LS09, DKLR11}]
For any nonsignaling distribution $\bp\in\C$,
$\R_\epsilon(\bp) +1 \geq \log(\nu_\epsilon(\bp)), $
and for any Boolean function $f$,
$\R_\epsilon(f)  \geq \log(\nu_\epsilon(\bp_f)).$
\end{proposition}

As lower bounds on communication complexity of Boolean functions go,
$\nu$ is one of the weaker bounds, equivalent to the smooth discrepancy~\cite{JK10}, and no larger than the approximate nonnegative rank and the
smooth rectangle bounds~\cite{KMSY14}.
More significantly for this work, up to small multiplicative constants, for Boolean functions,
(the log of) $\nu$ is
a lower bound on quantum communication, so
it is useless
to establish  gaps between classical and quantum communication complexity.
(This limitation, with the upper bound in terms of the number
of outputs
on normalized Bell violations, is a  consequence of
Grothendieck's theorem~\cite{Gro96}.)

The classical and quantum efficiency bounds, given here in their dual formulations, are expressed by
the following two convex optimization programs.
The classical bound is a generalization to distributions of the
partition bound of communication complexity~\cite{JK10,LLR12}.
This bound is one of the strongest lower bounds known, and can be exponentially larger than $\nu$ (an example is the Vector in Subspace problem~\cite{KR11}). It is always at least as large as the relaxed partition bound
which is in turn always at least as large as the smooth rectangle bound~\cite{JK10,KLLRX15}. Its weaker variants have been used to show exponential gaps between classical and quantum communication complexity.
The definition we give here is a stronger formulation than the one given in~\cite{LLR12}. We show they are equivalent in Appendix~\ref{app:equivEff}.
\begin{definition}[\cite{LLR12}]
\label{def:eff}
The $\epsilon$-error efficiency bound of a distribution $\bp\in\P$ is given by
\begin{align*}
 \eff_\epsilon(\bp)=&\max_{B,\beta} && \beta\\
&\textrm{subject to } &&
B(\bp')\geq \beta & \quad \forall \bp'\in\P\>\rm{\>s.t.\>}|\bp'-\bp|_1\leq\epsilon, \\
&&&B(\bl)\leq 1 & \quad \forall \bl\in\LL_{det}^\bot.
\end{align*}
We call any Bell functional that satisfies the second constraint in the above program {\em inefficiency-resistant Bell functional}.
The $\epsilon$-error quantum efficiency bound of a $\bp\in\P$ is
\begin{align*}
\eff_\epsilon^*(\bp)=&\max_{B,\beta}&&  \beta\\
&\textrm{subject to } &&
B(\bp')\geq \beta & \quad \forall \bp'\in\P\>\rm{\>s.t.\>}|\bp'-\bp|_1\leq\epsilon, \\
&&&B(\bq)\leq 1 & \quad \forall \bq\in\Q^\bot.
\end{align*}
We denote $\eff = \eff_0$ and $\eff^* = \eff^*_0$ the $0$-error bounds.
\end{definition}

Although it may not be straightforward from the above definition due to the presence of absolute values, the program for the classical efficiency bound is linear, as a consequence of $\LL_{det}^\bot$ being a polytope (See Appendix \ref{app:equivEff} and Ref.~\cite{LLR12}).

For any given distribution $\bp$, its classical communication complexity is bounded from below
by the (log
of the) efficiency bound. For randomized communication complexity with error $\epsilon$,
the bound is $\log(\eff_\epsilon)$ and for quantum communication complexity, the bound is
$\log(\eff^*_\epsilon$).
Note that for any $\bp\in \Q$, the quantum communication complexity is 0 and the
$\eff^*$ bound is 1.
For any function $f$, the efficiency bound $\eff_\epsilon(\bp_f)$ is equivalent to
the partition bound~\cite{JK10,LLR12}.

\begin{proposition}[\cite{LLR12}]
\label{prop:effLB}
For any 
$\bp \in \P$ and any $0\leq\epsilon<1/2$,
$ \R_\epsilon(\bp) \geq \log(\eff_\epsilon(\bp)) 
$
and
$Q_\epsilon^*(\bp) \geq \frac{1}{2}\log(\eff^*_\epsilon(\bp)).$
For any 
$\bp\in\C$ and any $0\leq\epsilon\leq 1$,
$\nu_\epsilon(\bp) \leq 2 \eff_\epsilon(\bp).$
\end{proposition}


Theorem~\ref{thm:existsBqeff} in Section \ref{sec:violations} below involves upper bounds on the quantum efficiency bound.
To give an upper bound on the quantum efficiency bound of a distribution $\bp$,
it is more convenient to use the primal formulation,
and upper bounds can be  given by
exhibiting a local (or quantum) distribution with abort which satisfies
the following two properties: the probability of aborting should be the same on
all inputs $x,y$, and conditioned on not aborting, the outputs of the protocol should reproduce the distribution $\bp$. The efficiency bound is inverse proportional to the probability of not aborting, so the goal is to abort as little as possible.
\begin{proposition}[\cite{LLR12}]
\label{prop:qeffprimal}
For any distribution $\bp\in\P$,
$\eff^*(\bp)= 1/\eta^*$, with $\eta^*$ the optimal value of the following optimization problem
(nonlinear, because $\Q^\bot$ is not a polytope).
\begin{align*}
   \max_{\bq\in \Q^\bot,\zeta}  \quad&\zeta \\
 \text{subject to}\quad & q(a,b|x,y)=\zeta p(a,b|x,y)
		 \quad \forall x,y,a,b\in \X{\times}\Y{\times}\A {\times}\B\\
\end{align*}
Moreover, for any $0\leq\epsilon\leq 1$, $\eff_\epsilon^*(\bp)= \min_{\bp'\in\P:|\bp'-\bp|_1\leq \epsilon} \eff^*(\bp')$.
\end{proposition}

\section{Properties of Bell inequalities}

Syntactically, there are two differences between the normalized Bell functionals (Definition~\ref{def:nu}) and the inefficiency-resistant ones (Definition~\ref{def:eff}).
The first difference is that the normalization constraint is relaxed: for inefficiency-resistant functionals, the Bell value for local distributions is, unlike the case for normalized Bell functionals, not bounded in absolute value; it is only bounded from above. Since this is a maximization problem, this relaxation
allows for larger violations.
This difference alone would not lead to a satisfactory measure of nonlocality, since one could obtain unbounded
violations by shifting and dilating the Bell functional.
The second difference prevents this.
The upper bound is required to hold for a larger set of local distributions, those that can abort. This is a much
stronger condition. Notice that a local distribution can selectively
abort on configurations that would otherwise tend to
keep the Bell value small, making it harder to satisfy the constraint.

In this section, we show that normalized Bell violations can be
modified to be resistant to local distributions that abort, while preserving the
violation on any nonsignaling distribution, up to a multiplicative factor of 3 (plus a term independent of the input and output sizes).
This means that we can add the stronger constraint of resistance to
local distributions that abort to Definition~\ref{def:nu},
incurring a loss of essentially a factor of 3,
and the only remaining difference between the resulting linear programs
is the relaxation of the lower bound (dropping the absolute value)
for local distributions that abort.

\begin{theorem}
\label{thm:BellLLbot}
Let $B$ be a normalized Bell functional on $\A\times\B\times\X\times\Y$ and $\bp\in\C$ a nonsignaling distribution such that $B(\bp)\geq 1$.
Then there exists a normalized Bell functional $B^*$ on $(\A\cup\{\bot\})\times(\B\cup\{\bot\})\times\X\times\Y$ with $0$ coefficients on
the $\bot$ outputs such that :
$\forall \bp \in \NS$,
  $B^*(\bp)\geq\frac{1}{3}B(\bp)-\frac{2}{3},$
  and
 $\forall \ell \in \LL_\det^\bot$,
  $|B^*(\ell)|\leq 1.$
\end{theorem}

The formal proof of Theorem~\ref{thm:BellLLbot} is deferred to Appendix~\ref{app:proofBellLLbot}, and we will only give its high-level structure in this part of the paper.
First, we show (see Observation~\ref{rem:saturatedBellLP}) how to rescale
a normalized Bell functional so that it saturates its normalization
constraint. Then, Definition~\ref{def:aborttomarginals} adds weights
to abort events to make the Bell functional resistant to inefficiency.
Finally, Lemma~\ref{lem:removeabortinB} removes the weights on the abort events of a Bell functional while keeping it bounded on the local set with abort, without dramatically changing the values it takes on the nonsignaling set. Our  techniques are similar to the ones used in~\cite{MPR+02}.

{\color{blue}
}

\section{Exponential violations from communication bounds}\label{sec:violations}

Recently, Buhrman {\it et al.} gave a general construction to
derive normalized Bell inequalities from any sufficiently large
gap between classical and quantum communication complexity.
\begin{theorem}[\cite{BCG+16}]
\label{thm:Buhrmanetal}
For any function $f$ for which there is a quantum protocol using $q$ qubits
of communication but no prior shared entanglement,
there exists a quantum distribution $\bq\in\Q$ and a normalized
Bell functional $B$ such that
$
 B(\bq) \geq \frac{\sqrt{R_{1/3}(f)}}{ 6\sqrt{30}q } (1-2^{-q})^{2q}.
$
\end{theorem}
Their construction is quite involved, requiring protocols to be memoryless,
which they show how to achieve in general, and
uses multiport teleportation to construct a quantum distribution.
The Bell inequality 
they construct expresses a correctness constraint.

In this section, we show how to obtain large inefficiency-resistant
Bell violations for quantum distributions from gaps between quantum communication complexity and the efficiency bound for classical communication complexity.
We first prove the stronger of two statements, which gives violations of $\frac{\eff_\epsilon(\bp)}{\eff_{\epsilon'}^\star(\bp)}$. 
For any problem for which
a classical lower bound $c$ is given using the efficiency or partition bounds
or any weaker method (including the rectangle bound and its variants), and
any upper bound $q$ on quantum communication complexity, it implies
a violation of $2^{c-2q}$.


\begin{theorem}\label{thm:existsBqeff}
For any distribution $\bp\in\P$ and any $0\leq \epsilon' \leq \epsilon \leq 1$,
if $(B,\beta)$ is a feasible solution to the dual of $\eff_\epsilon(\bp)$
and $(\zeta,\bq)$ is a feasible solution to the primal for $\eff^\star_{\epsilon'}(\bp)$,
then there is a quantum distribution $\overline{\bq}\in\Q$ such that
$
B(\overline{\bq}) \geq \zeta\beta$
and $B(\bl)\leq 1,  \forall \bl \in \LL_{det}^{\bot} \,,
$
and in particular, if both are optimal solutions, then
$
B(\overline{\bq}) \geq \frac{\eff_\epsilon(\bp)}{\eff_{\epsilon'}^\star(\bp)} \,.
$
The distribution $\overline{\bq}$ has one additional output per player compared to the distribution $\bp$.
\end{theorem}
\begin{proof}
Let $(B,\beta)$ be a feasible solution to the dual of $\eff_\epsilon(\bp)$,
$\bp'$ be such that $\eff^\star_{\epsilon'}(\bp)=\eff^\star(\bp')$ with $|\bp'-\bp|_1\leq\epsilon'$,
and $(\zeta,\bq)$ be a feasible solution to the primal for
$\eff^\star(\bp')$.
From the constraints, we have $ \bq\in \Q^\bot,$ $q(a,b|x,y)=\zeta p'(a,b|x,y)$ for all $(a,b,x,y)\in\A\times\B\times\X\times\Y,$ $B(\bl)\leq 1$ for all $\bl \in \LL_{det}^{\bot},$ and $B(\bp'')\geq\beta$ for all $\bp'' \textrm{ s.t. } |\bp''-\bp|_1\leq\epsilon.$
Then $B(\bq) =  \zeta B(\bp') \geq \zeta \beta $.
However, $\bq \in \Q^\bot$ but technically we want
a distribution in $\Q$ (not one that aborts). So
we add
a new (valid) output `A' to the set of outputs of each player, and they should
output `A' instead of aborting whenever~$\bq$ aborts.
The resulting distribution, say $\overline{\bq}\in\Q$ (with additional outcomes
`A' on both sides),
is such that $B(\overline{\bq})=B(\bq)$ (since the Bell functional $B$ does not have any weight on $\bot$ or on `A').
\end{proof}

Theorem~\ref{thm:Buhrmanetal} and
Theorem~\ref{thm:existsBqeff} are both general constructions, but there are
a few significant differences. Firstly, Theorem~\ref{thm:existsBqeff}
requires a lower bound on the partition bound in the numerator, whereas Theorem~\ref{thm:Buhrmanetal} only requires a lower bound on communication complexity (which could be exponentially larger). Secondly, Theorem~\ref{thm:Buhrmanetal} requires a
quantum communication protocol in the denominator, whereas our theorem only requires an upper bound on the quantum efficiency bound (which could be exponentially smaller). Thirdly, our bound is exponentially larger than Buhrman {\it et al.}'s  for most problems considered here, and applies to subquadratic gaps,
but their bounds are of the more restricted class of normalized Bell inequalities. These differences imply, in principle, different applicability regimes for the two constructions, which we depict in Figure \ref{fig:venn-diagram}.

Theorem~\ref{thm:existsBqeff}
gives an explicit Bell functional provided an
explicit solution to the efficiency (partition) bound is given
and the quantum distribution is obtained from a solution to the primal
of $\eff^\star$  (Proposition~\ref{prop:qeffprimal}).
Recall that a solution to the primal of
$\eff^\star$ is provided by a quantum
zero-communication protocol that can abort,
which conditioned on not aborting, outputs following~$\bp$.
We can also start from a quantum protocol, as we show below.
From the quantum protocol, we derive a quantum distribution using standard
techniques.

\begin{figure}[h]
  \centering
  \includegraphics[width=9cm]{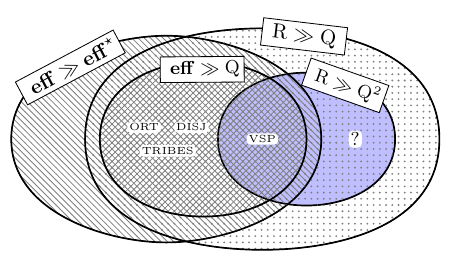}
  \caption{An illustration of the different applicability regions of Theorem~\ref{thm:Buhrmanetal} and Theorem~\ref{thm:existsBqeff}. There are problems in the literature for which both constructions apply, e.g. $\VSP$, and problems for which only our construction applies, e.g. $\DISJ$, $\ORT$, and $\TRIBES$ (see Table \ref{table:summaryResults}). Whether there exists any problem in the region in which $R\gg Q^2$ but $\eff\approx\eff^*$ (indicated by a question mark in the picture), for which Theorem~\ref{thm:Buhrmanetal} applies but our Theorem~\ref{thm:existsBqeff} does not, is, to our current knowledge, an open question.}
  \label{fig:venn-diagram}
\end{figure}

\begin{corollary}
\label{cor:existsBq}
For any distribution $\bp\in\P$ and any $0\leq\epsilon' \leq \epsilon  \leq 1$ such that
$R_\epsilon(\bp) \geq \log(\eff_\epsilon(\bp)) \geq c$ and
$Q^*_{\epsilon'}(\bp) \leq q $,
there exists an explicit inefficiency-resistant $B$ derived from
the efficiency lower bound, and an explicit quantum distribution $\overline{\bq}\in\Q$ derived from
the quantum protocol such that
$B(\overline{\bq}) \geq  2^{c-2q}$.
\end{corollary}

\begin{proof}
Let $(B,\beta)$ be an optimal solution to $\eff_\epsilon(\bp)$ and let $c$ be such that $\eff_\epsilon(\bp)=\beta \geq 2^c$. By optimality of $B$, we have  $B(\bp') \geq 2^c$ for any $\bp'$ such that $|\bp'-\bp|_1\leq\epsilon$.
Since $Q^*_{\epsilon'}(\bp) \leq q $, there exists a $q$-qubit quantum protocol (possibly using preshared entaglement) for some distribution $\bp'$ with $|\bp'-\bp|_1\leq\epsilon'\leq\epsilon$.
Then, we can use teleportation to obtain a $2q$ classical bit, entanglement-assisted protocol for $\bp'$.
We can simulate it without communication by picking a shared $2q$-bit random string and running the protocol but without
sending any messages.
If the measurements do not match the string,
output a new symbol `A' (not in the output set of the quantum protocol and different from $\bot$).
We obtain a quantum distribution $\overline{\bq}$ such that $B(\overline{\bq})= B(\bp')/2^{2q} \geq 2^{c-2q}$.
\end{proof}

Most often, communication lower bounds are not given as efficiency or partition bounds, but rather using variants of the corruption bound. We show in Section~\ref{sec:corruption} how to map a corruption bound to explicit Bell coefficients.

\section{Noise-resistant violations from communication bounds}

Normalized Bell inequalities are naturally resistant to
any constant level of local noise in the sense that the value of the Bell functional can only decrease by constant additive and multiplicative terms. Indeed, if the observed distribution is
$\tilde{\bp}=(1-\eps) \bp+\eps \bl$ for some $\bl\in\LL$, then
$B(\tilde{\bp}) \geq
(1-\eps) B(\bp) - \eps$ since $\abs{B(\bl)}\leq 1$.
In inefficiency-resistant Bell inequalities, relaxing the absolute value
leads to the possibility that $B(\bl)$ has a large negative value for
some local $\bl$. (Indeed, such large negative values are inherent to
large gaps between $\nu$ and $\eff$.)
If this distribution were used as adversarial noise,
the observed distribution, $(1-\eps) \bp+\eps \bl$, could have a  Bell value much smaller than $B(\bp)$.
This makes inefficiency-resistant Bell inequalities
susceptible to adversarial local noise.



In Theorem \ref{thm:uniform} below, whose proof we defer to Appendix \ref{app:noiseGeneral}, we show that the quantum violations arising from our construction in Corollary \ref{cor:existsBq} are resistant to the addition of white noise. More specifically, we consider the standard noise model (see e.g., \cite{MPR+02,KGZ+00,DKZ01}) in which the (ideal) quantum state $\ket{\psi}$ gets added a maximally mixed state with some probability $\delta$,
$$\rho := (1-\delta)\ket{\psi}\bra{\psi}+\delta\frac{\mathbb{I}}{d}.$$

We show that even in this noisy scenario, the inefficiency-resistant Bell functional $B$ and the quantum distribution $\bq$ resulting from the construction in Corollary \ref{cor:existsBq} are such that $B(\bq)\geq 2^{c-2q}$. Intuitively, resistance to this kind of noise follows from: 1) $B$ arising from a lower bound on bounded-error communication complexity and 2) the teleportation measurements needed to construct $\bq$ still giving uniform outcomes when we replace a maximally entangled with a maximally mixed state.

\begin{theorem}\label{thm:uniform}
For any $\bp\in\P$ and any $0\leq\epsilon' \leq \epsilon \leq 1$ such that
$\log(\eff_\epsilon(\bp)) \geq c$ and
$Q_{\epsilon'}^*(\bp) \leq q $, there exists an explicit inefficiency-resistant $\overline{B}$, POVMs $\{E_{a|x}\}_{a\in\A\cup \{S\}}$ over a Hilbert space $\mathcal{H_A}$ and $\{E_{b|y}\}_{b\in\B\cup \{S\}}$ over a Hilbert space $\mathcal{H_B}$ and a quantum state $\ket{\psi}\in\mathcal{H_A}\otimes \mathcal{H_B}$, such that for any $0\leq \delta\leq\epsilon-\epsilon'$, we have that
$$\overline{B}(\bq_\delta) \geq  2^{c-2q},$$
with
$$q_\delta(a,b|x,y)=\tr\left[E_{a|x}\otimes E_{b|y}\left((1-\delta) \ket{\psi}\bra{\psi}+\delta\frac{\mathbb{I}}{\dim({\mathcal{H_A}\otimes \mathcal{H_B}})}\right)\right].$$
Notice that $\bq$ has one additional output per player (labelled $S$) compared to $\bp$.
\end{theorem}

\section{Explicit constructions}

\subsection{From corruption bound to Bell inequality violation}
\label{sec:corruption}

We now explain how to construct an explicit Bell inequality violation from the corruption bound. The corruption bound, introduced by Yao in~\cite{Yao83}, is a very useful lower bound technique. It has been used for instance in~\cite{Razborov92} to get a tight $\Omega(n)$ lower bound on the randomized communication complexity of Disjointness (whereas the approximate rank, for example, can only show a lower bound of $\Theta(\sqrt{n})$).
Let us recall that a rectangle $R$ of $\X\times\Y$ is a subset of that set of the form $R_A\times R_B$, where
$R_A\subseteq\X$ and $R_B\subseteq\Y$.

\begin{theorem}[Corruption bound \cite{Yao83,babai1986complexity,KN97}]\label{thm:corruptionBoundCC}
Let $f$ be a (possibly partial) Boolean function on $\X\times\Y$. Given $\gamma, \delta \in (0,1)$, suppose that there is a distribution $\mu$ on $\X\times\Y$ such that for every rectangle $R \subseteq \X\times\Y$
\begin{equation*}
\mu(R\cap f^{-1}(1)) > \gamma \mu(R \cap f^{-1}(0)) - \delta
\end{equation*}
Then, for every $\epsilon \in (0,1)$,
$
        2^{R_{\epsilon}(f)} \geq \frac 1 \delta \left(\mu(f^{-1}(0)) - \frac \epsilon \gamma \right).
$
\end{theorem}
See, e.g., Lemma 3.5 in~\cite{Beame2006} for a rigorous treatment.
For several problems, such a $\mu$ is already known. In Theorem~\ref{thm:boundDistribBellGeneral} below, whose proof we defer to Appendix~\ref{app:corruption}, we show how to construct a Bell inequality violation from this type of bound.

\begin{theorem}\label{thm:boundDistribBellGeneral}
Let $f$ be a (possibly partial) Boolean function on $\X\times\Y$, where $\X,\Y\subseteq\{0,1\}^n$. Fix $z \in\ZO$. Let $\mu$ be an input distribution, and $(U_i)_{i \in I}$ (resp. $(V_j)_{j \in J}$) be a family of pairwise nonoverlapping subsets of $f^{-1}(\bar z)$ (resp. of $f^{-1}(z)$). Assume that there exists $g: \ints \to (0,+\infty)$ such that, for any rectangle $R \subseteq \X\times\Y$
\begin{equation}\label{eq:conditionDistribBound}
 \sum_{i \in I} u_i \mu(R \cap U_i) \geq \sum_{j \in J} v_j \mu(R \cap V_j) - g(n).
\end{equation}
Then, the Bell functional $B$ given by the following coefficients: for all $a,b,x,y \in \ZO\times\ZO\times\X\times\Y$,
\begin{equation}\label{eq:defGeneralB}
B_{a,b,x,y}=\begin{cases}
1/2(-u_i g(n)^{-1}\mu(x,y)) & \mbox{ if } (x,y)\in U_i \text{ and } a\oplus b=z, \\
1/2(v_j g(n)^{-1}\mu(x,y)) & \mbox{ if } (x,y)\in V_j \text{ and } a\oplus b=z, \\
0 & \mbox{otherwise.}
\end{cases}
\end{equation}
satisfies
\begin{align}
B(\bl)&\leq 1, \quad  \forall \bl \in \LL_{det}^{\bot},\label{eq:corruption-BellInequality-local-bound}\\
B(\bp_f) &= \frac{1}{2\cdot g(n)}\sum_{j}{v_j\mu(V_j)}\label{eq:corruption-BellInequality-Pf}
\end{align}
and for any $\bp'\in\P$ such that $|\bp'-\bp_f|_1\leq \epsilon$ :
\begin{equation}\label{eq:corruption-BellInequality-Pfprime}
        B(\bp') \geq \frac{1}{2\cdot g(n)}\left [ \sum_{j}{v_j\mu(V_j)} - \epsilon \left(\sum_{j}{|v_j|\mu(V_j)} +\sum_{i}{|u_i|\mu(U_i)}\right) \right ].
\end{equation}
\end{theorem}

For many other problems in the literature, such as Vector in Subspace and Tribes, stronger variants of the corruption bound are needed to obtain good lower bounds. These stronger variants have been shown to be no stronger than
the partition bound (more specifically, the relaxed partition bound)~\cite{KLLRX15}. The generalization in Theorem~\ref{thm:boundDistribBellGeneral} of the hypothesis of Theorem~\ref{thm:corruptionBoundCC}, which the reader might have noticed, allow us to construct explicit Bell functionals also for these problems.

\subsection{Some specific examples}\label{subsec:examples}


Using Corollary~\ref{cor:existsBq} and the construction to go
from a corruption bound (or its variants) to a Bell inequality (Theorem~\ref{thm:boundDistribBellGeneral}), we
give explicit Bell inequalities and violations for several problems
studied in the literature.
Since our techniques also  apply to small gaps,
we include problems for which the
gap between classical and quantum communication complexity is polynomial.

\paragraph*{Vector in Subspace}
%
In the Vector in Subspace Problem $\VSP_{0,n}$, Alice is given an $n/2$ dimensional subspace of an $n$ dimensional space over $\Real$, and Bob is given a vector. This is a partial function, and the promise is that either Bob's vector lies in the subspace, in which case the function evaluates to $1$, or it lies in the orthogonal subspace, in which case the function evaluates to $0$. Note that the input set of $\VSP_{0,n}$ is continuous, but it can be discretized by rounding, which leads to the problem $\widetilde\VSP_{\theta,n}$ (see~\cite{KR11} for details).
Klartag and Regev~\cite{KR11} show that the VSP can be solved with an $O(\log n)$ quantum protocol, but the randomized communication complexity of this problem is $\Omega(n^{1/3})$. As shown in~\cite{KLLRX15}, this is also a lower bound on the relaxed partition bound.
Hence Corollary~\ref{cor:existsBq} yields the following.
\begin{proposition}
\label{cor:VSP}
There exists a Bell inequality $B$ and
a quantum distribution $\overline{\bq}_{VSP}\in\Q$
such that
$
B\left(\overline{\bq}_{VSP}\right) \in 2^{\Omega(n^{1/3}) - O(\log n)}
$
and
for all $\bl \in \LL_{det}^\bot,\>B(\bl) \leq 1$.
\end{proposition}
Note that the result of~\cite{KR11} (Lemma~4.3) is not of the form needed to apply Theorem~\ref{thm:boundDistribBellGeneral}. It is yet possible to obtain an explicit Bell functional following the proof of Lemma~5.1 in~\cite{KLLRX15}.

\paragraph*{Disjointness}
%
In the Disjointness problem, the players receive two sets and have to determine whether they are disjoint or not. More formally, the Disjointness predicate is defined over $\X = \Y = \mathcal{P}([n])$ by $\DISJ_n(x,y) = 1$ iff $x$ and $y$ are disjoint. It is also convenient to see this predicate as defined over length $n$ inputs, where $\DISJ_n(x,y) = 1$ for $x,y \in \{0,1\}^n$ if and only if $|\{i \,:\, x_i = 1 = y_i\}| = 0$.
The communication complexity for $\DISJ_n$ is $\Omega(n)$ using a corruption bound~\cite{Razborov92}
and there is a quantum protocol using $O(\sqrt{n})$ communication~\cite{AA05}.
Combining these results with ours, we obtain the following.
\begin{proposition}
	\label{cor:DisjBellIneq}
	There is a quantum distribution $\overline{\bq}_\DISJ\in\Q$ and
an explicit Bell inequality~$B$ satisfying: 	$
		B(\overline{\bq}_{\DISJ}) =  2^{\Omega(n) - O(\sqrt{n})}
	$,
	 and for all $\bl \in \LL_{det}^\bot,\>B(\bl) \leq 1$.
\end{proposition}
The proof is deferred to the Appendix (see~Section~\ref{subsec:DISJ}).

\paragraph*{Tribes.}
%
Let $r\geq 2$, $n=(2r+1)^2$. Let $\TRIBES_n:\{0,1\}^n\times\{0,1\}^n\to\{0,1\}$ be defined as:
$
\TRIBES_n(x,y) := \bigwedge\limits_{i=1}^{\sqrt{n}} \left( \bigvee\limits_{j=1}^{\sqrt{n}} (x_{(i-1)\sqrt{n}+j}\land y_{(i-1)\sqrt{n}+j})\right).
$
The Tribes function has an $\Omega(n)$ classical lower bound~\cite{HJ13} using the smooth rectangle bound and a $O(\sqrt{n}(\log n)^2)$ quantum protocol~\cite{BCW98}.
Combining these results with ours, we obtain the following.
\begin{proposition}
	\thlabel{coro:TribesBellIneq}
	There is
a quantum distribution  $\overline{\bq}_\TRIBES\in\Q$ and
an explicit Bell inequality~$B$ satisfying:
$
		B(\overline{\bq}_{\TRIBES}) = 2^{\Omega(n)-O(\sqrt{n}(\log n)^2)}
	$,
	and for all $\bl \in \LL_{det}^\bot,\>B(\bl) \leq 1$.
\end{proposition}
The proof is deferred to the Appendix (see~Section~\ref{subsec:TRIBES}).

\paragraph*{Gap Orthogonality.}
%
The Gap Orthogonality ($\ORT$)
problem was introduced by Sherstov as an intermediate step to prove a lower
bound for the Gap Hamming Distance ($\GHD$) problem~\cite{She12}. We derive an
explicit Bell inequality for $\ORT$ from Sherstov's lower bound of $\Omega(n)$,
shown in~\cite{KLLRX15} to be a relaxed partition bound. (Applying Corollary~\ref{cor:existsBq} also gives a (nonexplicit) violation for $\GHD$.)
The quantum upper bound is $O(\sqrt{n}\log n)$ by the general result of~\cite{BCW98}.
In the $\ORT$ problem, the players receive vectors and need to tell whether they are nearly orthogonal or far from orthogonal. More formally,
we consider the input space $\{-1,+1\}^n$ (to stick to the usual notations for this problem),
and we denote $\langle \cdot, \cdot \rangle$ the scalar product on $\{-1,+1\}^n$.
Let $\ORT_n : \{-1,+1\}^n\times\{-1,+1\}^n  \to \{-1,+1\}$ be the partial function defined as in~\cite{She12} by:
$\ORT_n(x,y) = -1$ if $|\langle x,y\rangle| \leq \sqrt{n},$ and
$\ORT_n(x,y) = +1$ if $|\langle x,y\rangle| \geq 2\sqrt{n}.$
Combining the results mentioned above with ours, we obtain the following.
\begin{proposition}
	\thlabel{coro:ORTBellIneq}
	There is 
a quantum distribution $\overline{\bq}_\ORT\in\Q$ and
an explicit Bell inequality $B$ satisfying:
	$
		B(\overline{\bq}_\ORT) = 2^{\Omega(n)-O(\sqrt{n}\log n)}
	$,
	and
	for all $\bl \in \LL_{det}^\bot,\>B(\bl) \leq 1$.
\end{proposition}
The proof is deferred to the Appendix (see~Section~\ref{subsec:ORT}).

\section{Discussion\label{sec:discussion}}

We have given three main results. First, we showed that normalized Bell inequalities can be modified to be bounded in absolute value on the larger set of local distributions that can abort without significantly changing the value of the violations achievable with nonsignaling distributions.
Then, we showed how to derive large inefficiency-resistant Bell violations from any gap between the partition bound and the quantum communication complexity of some given distribution $\bp$. The distributions $\bq$ achieving the large violations are relatively simple (only $3$ outputs for boolean distributions $\bp$) and are resistant to the addition of white noise to the ideal entangled state giving rise to them.
Finally, we showed how to construct explicit Bell inequalities when the separation between classical and quantum communication complexity is proven via the corruption bound.

From a practical standpoint, the specific Bell violations we have studied are probably
not feasible to implement, because the parameters needed are still impractical or the quantum states are infeasible to implement.
However, our results suggest that we could 
consider
functions with small gaps
in communication complexity, in order to find practical Bell inequalities that are robust against uniform noise and detector inefficiency.
%
As an example, let us show how we could design a Bell test from a gap in communication complexity using our techniques. Consider a boolean function $f$ with a communication complexity lower bound $R_\epsilon(\bp_f)\geq\log(\eff_\epsilon(\bp_f))\geq c(n)$, and accepting a quantum protocol using at most $q(n)$ qubits of communication with error at most $\epsilon'<\epsilon$, where $n$ is the size of the inputs.
The construction from Theorem~\ref{thm:uniform} then leads to a Bell inequality $\overline{B}$, a quantum state $\ket{\psi}\in\mathcal{H_A}\otimes \mathcal{H_B}$ and POVMs $\{E_{a|x}\}$ and $\{E_{b|y}\}$ over $\mathcal{H_A}$ and $\mathcal{H_B}$ respectively, such that, for any $\delta\leq \epsilon-\epsilon'$, $\overline{B}(\bq_\delta)\geq 2^{c(n)-2q(n)}$ for $\bq_\delta(a,b|x,y)=\tr\left[E_{a|x}\otimes E_{b|y}\left((1-\delta) \ket{\psi}\bra{\psi}+\delta\frac{\mathbb{I}}{\dim({\mathcal{H_A}\otimes \mathcal{H_B}})}\right)\right]$. More precisely, $\ket{\psi}$ is obtained by preparing $q(n)$ EPR pairs shared between Alice and Bob, and the POVMs $\{E_{a|x}\}$ and $\{E_{b|y}\}$ arise from the players performing the local unitaries described by the quantum protocol together with a total of $q(n)$ local Bell measurements (each qubit communicated by Alice to Bob in the quantum communication protocol for $\bp_f$ will turn into a Bell measurement on Alice's side, and vice versa on Bob's side) and a final measurement on both sides to obtain outcomes $(a,b)$ in $\bp_f$. Practically, these measurements will not always be successful due to the imperfect efficiency of the detectors. Suppose that the efficiency for each Bell measurement is $\eta_0$, and that the overall efficiency for the final measurements is at least $\eta_0^{2q(n)}$ (these final measurements involve $2q(n)$ qubits, which justifies such an expected scaling). Therefore, factoring in this efficiency and assuming that the noise affecting the experiment is $\delta<\epsilon-\epsilon'$, the observed Bell value will be $\tilde{B}(\bq')\geq 2^{c(n)-2q(n)}\eta_0^{2q(n)}$, hence we still observe a Bell violation (a value $\tilde{B}(\bq')>1$) as long as $n$ is chosen high enough (or the detector efficiencies are large enough) to have $\eta_0>2^{-(\frac{c(n)}{2q(n)}-1)}$.

The preceding analysis, although accounting for detectors whose efficiency decreases exponentially with the number of EPR measurements $q(n)$, assumes that the (uniform) $\delta$-noise affects the quantum state both globally and independently of $n$. This is a quite standard noise model, used in various studies of Bell inequalities such as \cite{MPR+02,KGZ+00,DKZ01}. In a practical scenario, one might nevertheless want to consider an arguably more realistic noise model where the error increases with the dimension of the global quantum state (this would happen for example in the case of noisy teleportation, where each individual EPR pair gets added a maximally mixed state). As is, our construction is not resistant to such a noise model as it is based on a classical lower bound for \emph{bounded-error} communication complexity, which can therefore only tolerate bounded noise. One way around this issue would be to turn the bounded-error classical lower bound into a bound valid for low success probabilities using standard amplification techniques (see, e.g. \cite{KN97}). Another possible solution would be to reduce the noise in the quantum protocol by using quantum error correction (see, e.g. \cite{NC02}). We leave these ideas for future work.

Lastly, we comment on upper bounds for the violation of inefficiency-resistant Bell inequalities. First, since (the log of) the efficiency bound is a lower bound on communication complexity, these violations are bounded above by the number of inputs per side.
%
%
%
Next, for dimension $d$ and number of outcomes $K$, we obtain the upper bound
$\eff_\epsilon(\bq) \leq 2^{O((\frac{Kd}{\epsilon})^2 \log^2(K))}$
for quantum distributions, by combining known bounds.
Indeed, we know that
$R_{\epsilon}(\bp) \leq O((\frac{K\nu(\bp)}{\epsilon})^2 \log^2(K))$ for any $\bp\in\NS$ (see~\cite{DKLR11}).
Combining this with the  bounds
$\eff_\epsilon(\bp)\leq 2^{R_\epsilon(\bp)}$~(Proposition~\ref{prop:effLB}), and
$\nu(\bq)\leq O(d)$
for any $\bq\in\Q$ (see~\cite{JPP+10}),  gives the desired upper bound.
Hence unbounded violations are possible for $K=3$ outputs per side.





\begin{acknowledgments}
 We would like to acknowledge the following sources of funding for this work:
 the European Union Seventh Framework Programme (FP7/2007-2013) under grant agreement no. 600700 (QALGO), the Argentinian ANPCyT (PICT-2014-3711), the Laboratoire International Associ\'ee INFINIS, the Belgian ARC project COPHYMA and the Belgian \textit{Fonds de la Recherche Scientifique} - FNRS under grant no. F.4515.16 (QUICTIME), and the French ANR Blanc grant RDAM ANR-12-BS02-005.
\end{acknowledgments}

\bibliographystyle{plainnat}

\appendix

\section{Proof of Theorem~\ref{thm:BellLLbot}}
\label{app:proofBellLLbot}

\begin{observation}
\label{rem:saturatedBellLP}
Let $B$ be a nonconstant normalized Bell functional and $\bp\in\C$ such that $B(\bp) \geq 1$. Consider $\bl^- \in \LL_\det^\bot$ such that $B(\bl^-) = m = \min \{B(\bl) | \bl \in \LL_\det^\bot\}$ and $\bl^+ \in \LL_\det^\bot$ such that $B(\bl^+) = M = \max \{B(\bl) | \bl \in \LL_\det^\bot\}$. We have $m < M$ because $B$ is nonconstant.
The Bell functional $\tilde{B}$ defined by
$
\tilde{B}(\cdot) =\frac{1}{M-m}(2B(\cdot)-M-m),
$
is such that $\tilde{B}(\bl^+)=1$, $\tilde{B}(\bl^-)=-1$, $|\tilde{B}(\ell)|\leq 1$ for all $\ell\in\LL_\det^\bot$, and $\tilde{B}(\bp) \geq B(\bp)$ since $B$ is normalized.
\end{observation}

Definition~\ref{def:aborttomarginals} below is the first step of the construction. It takes two marginal distributions $\bbm_A$ and $\bbm_B$, and a normalized Bell functional $B$,
and constructs a Bell functional $B_{\bbm_A,\bbm_B}^\bot$ whose value over every distribution $\bp \in \NS^\bot$ coincides with the value of $B$ over the distribution $\bp'\in\NS$ obtained from $\bp$ by replacing the abort events with samples from $\bbm_A$ and $\bbm_B$.

\begin{definition}
\label{def:aborttomarginals}
For all two families of distributions,
$\bbm_A=(m_A(\cdot|x))_{x\in\X}$ over outcomes in $\A$ for Alice and $\bbm_B=(m_B(\cdot|y))_{y\in\Y}$ over outcomes in $\B$ for Bob,
and any normalized Bell functional $B$ with coefficients only on nonaborting events, we define
the 
Bell functional $B^\bot_{\bbm_A,\bbm_B}$ on $(\A\cup\{\bot\})\times(\B\cup\{\bot\})\times\X\times\Y$ by
\begin{multline*}
(B^\bot_{\bbm_A,\bbm_B})_{a,b,x,y} =
B_{a,b,x,y} + \chi_{\{\bot\}}(a) \sum_{a'\neq\bot}m_A(a'|x)B_{a',b,x,y}
\\\quad+\chi_{\{\bot\}}(b) \sum_{b'\neq\bot}m_B(b'|y)B_{a,b',x,y} + \chi_{\{\bot\}}(a)\chi_{\{\bot\}}(b) \sum_{a',b'\neq\bot}m_A(a'|x)m_B(b'|y)B_{a',b',x,y}
\end{multline*}
where $\chi_{\mathcal{S}}$ is the indicator function for set $\mathcal{S}$ taking value $1$ on $\mathcal{S}$ and $0$ everywhere else.
\end{definition}

\begin{observation} Let $f_{\bbm_A,\bbm_B} : \NS^\bot \rightarrow \NS$ 
be the function that replaces abort events on Alice's (resp. Bob's) side 
by a sample from $\bbm_A$ (resp. $\bbm_B$) (note that $f_{\bbm_A,\bbm_B}$ preserves locality). Then,
for every $\bbm_A$, $\bbm_B$ and $B$ as in Definition~\ref{def:aborttomarginals}, the Bell
functional $B^\bot_{\bbm_A,\bbm_B}$ satisfies that $B^\bot_{\bbm_A,\bbm_B}(\bp)=B(f_{\bbm_A,\bbm_B}(\bp)),\ \forall \bp \in \NS^\bot$, so $B^\bot_{\bbm_A,\bbm_B}(\bp) =B(\bp),$ for all $\bp \in \NS,$ and $|B^\bot_{\bbm_A,\bbm_B}(\bl)| \leq 1,$ for all  $\bl \in \LL^\bot.$
\end{observation}

Next, in Lemma~\ref{lem:removeabortinB} below, we do without the abort coefficients in the Bell functionals $B_{\bbm_A,\bbm_B}^\bot$.

\begin{lemma}
\label{lem:removeabortinB}
Let $B'$ be a normalized Bell functional on $\A^\bot\times\B^\bot\times\X\times\Y$ (possibly with nonzero weights on $\bot$).
Then the Bell functional $B''$ on the same set defined by
\begin{equation}\label{eq:defBstarNobot}
                B''_{a,b,x,y} = B'_{a,b,x,y} -B'_{a,\bot, x,y} -B'_{\bot,b,x,y} + B'_{\bot,\bot, x,y},
\end{equation}
for all $(a,b,x,y)\in(\A\cup\{\bot\})\times(\B\cup\{\bot\})\times\X\times\Y$\\
satisfies :

\begin{enumerate}
\item\label{item:removeabortinB-1}
If $a=\bot\ or\ b=\bot$ then $B''_{a,b,x,y}=0$
\item\label{item:removeabortinB-2} for all $\bp \in \NS$,
\begin{equation}\label{eq:BppFomrBp}
	B''(\bp) = B'(\bp) - B'(\bp_{A,\bot}) - B'(\bp_{\bot,B}) + B'(\bp_{\bot,\bot}),
\end{equation}
\end{enumerate}
where $\bp_{A,\bot}\in\LL^\bot$ (resp. $\bp_{\bot,B}\in\LL^\bot$) is the local distribution obtained from $\bp$ if Bob (resp. Alice) replaces all of his (resp. her) outputs by $\bot$, and $\bp_{\bot,\bot}\in\LL^\bot$ is the local distribution where both Alice and Bob always output $\bot$. In Item 2 above, for all $\bp'$,
$$B'(\bp') = \sum_{(a,b) \in \A^\bot\times\B^\bot}  \sum_{(x,y) \in \X\times\Y}  B'_{a,b,x,y} \bp'(a,b|x,y)$$
where the first sum is also over the abort events.
\end{lemma}

\begin{proof}
Item~\ref{item:removeabortinB-1} follows from~(\ref{eq:defBstarNobot}). We prove Item~\ref{item:removeabortinB-2}.
For $\bp\in\NS^\bot$ with marginals $\bp_A$ and $\bp_B$, we have:
for all $y \in Y$, $p_A(a|x) = \sum_{b\in\B^\bot} p(a,b | x,y),$
and for all $x \in X$, $p_B(b|y) = \sum_{a\in\A^\bot} p(a,b | x,y).$
For the remainder of this proof, summations involving  $a$ (resp. $b$) are over $a\in\A^\bot$ (resp.~$b\in\B^\bot$).

By definition,
$p_{A,\bot}(a,b|x,y) = p_A(a|x)\chi_{\{\bot\}}(b),$
$p_{\bot,B}(a,b|x,y) = \chi_{\{\bot\}}(a) p_B(b|y),$
and $p_{\bot,\bot}(a,b|x,y) = \chi_{\{\bot\}}(a)\chi_{\{\bot\}}(b).$
We have:
\begin{align*}
                B''(\bp)
                &= \sum_{a,b,x,y} \left[B'_{a,b,x,y} -B'_{a,\bot, x,y}
         -B'_{\bot, b,x,y} + B'_{\bot,\bot ,x,y}\right] p(a,b | x,y) \\
                &= \sum_{a,b,x,y} B'_{a,b,x,y} p(a,b | x,y)
- \sum_{a,x,y}B'_{a,\bot ,x,y}\sum_{b}p(a,b | x,y)\\
                &\qquad\qquad - \sum_{b,x,y}B'_{\bot, b,x,y}\sum_{a}p(a,b | x,y)
                + \sum_{x,y}B'_{\bot,\bot ,x,y}\sum_{a,b} p(a,b | x,y) \\
                &= B'(\bp) - \sum_{a,x,y}B'_{a,\bot ,x,y}p_A(a|x)
- \sum_{b,x,y}B'_{\bot ,b,x,y}p_B(b|y) + \sum_{x,y}B'_{\bot,\bot, x,y} \\
                &= B'(\bp) - B'(\bp_{A,\bot}) - B'(\bp_{\bot,B}) + B'(\bp_{\bot,\bot}).
                \tag*{\qedhere}
      \end{align*}

\end{proof}

We are now ready to prove Theorem~\ref{thm:BellLLbot}.

\begin{proof}[Proof of Theorem~\ref{thm:BellLLbot}]
If $B$ is constant, since it is normalized by assumption, we have $B \equiv 1$. Thus, we can simply take $B^*$ defined by: for all $(x,y)\in\X\times\Y$, $B^*_{a,b,x,y} = B_{a,b,x,y}$ if $(a,b) \in \A\times\B$, and $B^*_{a,b,x,y} = 0$ otherwise.

Now, let us assume that $B$ is not constant and let $\bl^-,\bl^+ \in \LL_\det$, and $\tilde{B}$ constructed from $B$ as in Observation~\ref{rem:saturatedBellLP} satisfying $\tilde{B}(\bl^-)=-1$ and $\tilde{B}(\bl^+)=1$. %
Since $\bl^-$ and $\bl^+$ are deterministic distributions, we have: $\bl^- = \bl_A^- \otimes \bl_B^-$ and $\bl^+ = \bl_A^+ \otimes \bl_B^+$, for some marginals $\bl_A^- , \bl_B^-,\bl_A^+,$ and $ \bl_B^+$.
We consider the replacing Bell functional $B^\bot_{\bl_A^-,\bl_B^-}$ (resp. $B^\bot_{\bl_A^+,\bl_B^+}$) from Definition~\ref{def:aborttomarginals} constructed from $(\tilde{B},\bl_A^-, \bl_B^-)$ (resp. from $(\tilde{B},\bl_A^+, \bl_B^+)$). Taking $B' = \frac 1 2 (B^\bot_{\bl_A^-,\bl_B^-} + B^\bot_{\bl_A^+,\bl_B^+})$, we have $|B'(\bl)| \leq 1,$ for all $\bl \in \LL^\bot$, and therefore we can apply Lemma~\ref{lem:removeabortinB} to get~$B''$ from $B'$.
Since $B'(\bp_{\bot,\bot}) = \frac 1 2 (B^\bot_{\bl_A^-,\bl_B^-}(\bp_{\bot,\bot}) + B^\bot_{\bl_A^+,\bl_B^+}(\bp_{\bot,\bot})) = \frac 1 2 (\tilde{B}(\bl^-)+\tilde{B}(\bl^+)) = 0 $, by~\eqref{eq:BppFomrBp} we have
for all $
\bp \in \NS^\bot,$ $\ B''(\bp) = B'(\bp) - B'(\bp_{A,\bot}) - B'(\bp_{\bot,B}).
$
Hence, denoting $B^*=\frac 1 3 B''$, $B^*$ satisfies all the required properties since $|B'(\bl)| \leq 1$ for all $\bl \in \LL^\bot$ and therefore we have for all $\bp\in \NS,$
$
B^*(\bp) \geq \frac{1}{3}B'(\bp) - \frac{1}{3} |B'(\bp_{A,\bot})| - \frac{1}{3} |B'(\bp_{\bot,B})|\geq\frac{1}{3}B'(\bp) - \frac{2}{3},
$
and for all $\bl \in \LL^\bot,$
$
|B^*(\bl)| \leq \frac{1}{3}|B'(\bl)| + \frac{1}{3} |B'(\bl_{A,\bot})| + \frac{1}{3} |B'(\bl_{\bot,B})|\leq 1.
$
\end{proof}

\section{Proof of Theorem \ref{thm:uniform}}
\label{app:noiseGeneral}
\begin{proof}
Let $(B,\beta)$ be an optimal solution to $\eff_\epsilon(\bp)$ and let $c$ be such that $\eff_\epsilon(\bp)=\beta \geq 2^c$. Recall that, by optimality of $B$, we have
\begin{align}
B(\bp') \geq 2^c\mbox{ for any }\bp'\mbox{ such that }|\bp'-\bp|_1\leq\epsilon.\label{eqn:optim-sol-eff}
\end{align}
As in the proof of Corollary \ref{cor:existsBq}, we go
from a $q$-qubit quantum protocol for a distribution $\bp'$ with $|\bp'-\bp|_1\leq\epsilon'$ to
a $2q$-bit entanglement-assisted protocol for $\bp'$ by replacing quantum communication with teleportation.
Let $\mathcal{H_{A'}}$ (resp. $\mathcal{H_{B'}}$) be Alice's (resp. Bob's) local Hilbert spaces in this protocol and, without loss of generality, let the initial state of the protocol be $$\ket{\psi}:=\ket{\beta_{00}}^{\otimes q}\ket{\phi}$$ with $\ket{\beta_{00}}:=1/\sqrt{2}({\ket{00}+\ket{11}})$ and some fixed $\ket{\phi}\in \mathcal{H_{A'}}\otimes\mathcal{H_{B'}}$.
Next, we go to a zero-communication protocol $\Pi$ by instructing the players not to communicate the teleportation
measurements' outcomes (and hence not to perform the correcting unitaries) and to output a new symbol $S\notin \A\cup\B\cup\{\bot\}$ whenever they get a teleportation measurement's outcome different from $00$ (the outcome corresponding to the case in which no correcting unitary is required); notice that this happens with probability $1-2^{-2q}$. The local unitaries and local measurements performed by the players
during the execution of $\Pi$ effectively induce POVMs $\{E_{a|x}\}_{a\in\A\cup \{S\}}$ and $\{E_{b|y}\}_{b\in\B\cup \{S\}}$ over $\mathcal{H_A}:=\mathcal{H_{A'}}\otimes \mathbb{C}^{2^q}$ and $\mathcal{H_B}:=\mathbb{C}^{2^q}\otimes \mathcal{H_{B'}}$ respectively and a corresponding quantum distribution $\bq$ such that the probability of the players outputting $(a,b)$ on inputs $(x,y)$ is given by $q(a,b|x,y)=\tr\left[ (E_{a|x}\otimes E_{b|y}) \ket{\psi}\bra{\psi}\right]$. Notice that
\begin{align}
q(a\neq S,b\neq S|x,y)=\tr\left[(\mathbb{I}-E_{S|x})\otimes(\mathbb{I}-E_{S|y})\ket{\psi}\bra{\psi}\right]=2^{-2q},
\label{eqn:prob-teleport}
\end{align}
which together with the fact that we started from a quantum protocol for $\bp'$, implies that $\bq$ is of the form
$$
q(a,b|x,y)=\frac{1}{2^{2q}}p'(a,b|x,y)+(1-\frac{1}{2^{2q}})r_1(a,b|x,y),
$$
with $\mathbf{r_1}$ supported on events in which at least one of the players output $S$. Therefore, letting $\overline{B}$ be the Bell functional such that $\overline{B}_{a,b,x,y}=B_{a,b,x,y}$ if $a\neq S$ and $b\neq S$, and $\overline{B}_{a,b,x,y}=0$ otherwise, we have that
$$\overline{B}(\bq)= \frac{1}{2^{2q}}B(\bp') \geq 2^{c-2q}.$$

Next, suppose that we run protocol $\Pi$ with the initial state
$$
\rho := (1-\delta) \ket{\psi}\bra{\psi}+\delta\left(\frac{\mathbb{I}}{2^{2q}}\otimes\frac{\mathbb{I}}{\dim(\mathcal{H_{A}'})\dim(\mathcal{H_{B}'})}\right),
$$
and let $\bq_\delta$ be the resulting quantum distribution. Notice that the $2q$ teleportation measurements will still give uniform outcomes when performed with $\mathbb{I}/4$ instead of $\ket{\beta_{00}}$, i.e. \eqref{eqn:prob-teleport} still holds when we replace $\ket{\psi}$ with $\rho$, and hence we have that
\begin{align*}
q_\delta(a,b|x,y)&=\tr\left[(E_{a|x}\otimes E_{b|y})\rho\right]\\
&=(1-\delta)\left[\frac{1}{2^{2q}}p'(a,b|x,y)+(1-\frac{1}{2^{2q}})r_1(a,b|x,y)\right]+\\
&\qquad\qquad\qquad\delta \left[\frac{1}{2^{2q}}n(a,b|x,y)+(1-\frac{1}{2^{2q}})r_2(a,b|x,y)\right]\\
&=\frac{1}{2^{2q}}\left[(1-\delta)p'(a,b|x,y)+ \delta n(a,b|x,y)\right]+(1-\frac{1}{2^{2q}})\left[(1-\delta) r_1(a,b|x,y)+\delta r_2(a,b|x,y)\right]
\end{align*}
for distributions $\mathbf{r_2}$ supported on events with at least one of the players outputting $S$ and $\mathbf{n}$ on events in which none of the players output $S$.

Finally, leting $\bp'':=(1-\delta)\bp'+\delta\mathbf{n}$, and noticing that
\begin{align*}
|\bp-\bp''|_1&=|(1-\delta)\bp'+\delta\mathbf{n}|_1\\
&\leq \epsilon'+\delta |\bp'-\mathbf{n}|_1\\
&\leq \epsilon'+\delta\leq \epsilon,
\end{align*}
with the last inequality following from $\delta\leq \epsilon-\epsilon'$, we get from \eqref{eqn:optim-sol-eff} that,
\begin{align*}
\overline{B}(\bq_\delta)=\frac{1}{2^{2q}}B(\bp'')\geq 2^{c-2q}.
\end{align*}
\end{proof}

\section{Proof of Theorem~\ref{thm:boundDistribBellGeneral}}
\label{app:corruption}

\begin{proof}
Let us first set $B_{z,x,y}=B_{a,b,x,y}$ for all $a\oplus b=z$. Let $\bl\in\LL_{det}^{\bot}$. Then, we have:
\begin{equation*}
B(\bl)=\sum_{(x,y)\in R}B_{z,x,y}+\sum_{(x,y)\in S}B_{z,x,y}
\end{equation*}
where $R$ and $S$ are the two rectangles where $\bl$ outputs $z$.
Let us take a rectangle $R$. Then :
\begin{align*}
\sum_{(x,y)\in R}B_{z,x,y}=& \frac{1}{2\cdot g(n)} \left(\sum_j v_j\mu(V_j\cap R) - \sum_i u_i\mu(U_i\cap R)\right)
\leq 1/2
\end{align*}
with the inequality following from (\ref{eq:conditionDistribBound}). This proves (\ref{eq:corruption-BellInequality-local-bound}).

Let us now compute $B(\bp_f)$. By linearity of $B$ and the definition of its coefficients, we have:
\begin{align*}B(\bp_f)  &=\sum_{a,b,x,y}B_{a,b,x,y}\bp_f(a,b|x,y)                                      \\
                                        &=\frac{1}{2}\sum_{(x,y) \in f^{-1}(z), a,b} B_{a,b,x,y} \chi_{\{z\}}(a \oplus b)+\frac{1}{2}\sum_{(x,y) \in f^{-1}(\bar z), a,b} B_{a,b,x,y} \chi_{\{\bar{z}\}}(a \oplus b)             \\
                                        &=1/2\sum_{j}\sum_{(x,y)\in V_j}v_j g(n)^{-1}\mu(x,y)      \\
                                        &=\frac{1}{2\cdot g(n)}\sum_{j}{v_j\mu(V_j)}
\end{align*}
(for the third equality we used the fact that $B_{a,b,x,y} = 0$ when $a \oplus b =  \bar z$). This proves \eqref{eq:corruption-BellInequality-Pf}.

Moreover, for any family of additive error terms $\Delta(a,b|x,y)\in[-1,1]$ such that
\begin{align*}
\sum_{a,b}|\Delta(a,b|x,y)|&\leq \epsilon & \forall x,y\in\X\times\Y,
\end{align*}
denoted collectively as $\bdelta$, we have
\begin{align*}
        |B(\bdelta)| &= \left|\sum_{a,b,x,y}B_{a,b,x,y}\Delta(a,b|x,y)\right| \\
                                        &= \frac{1}{2\cdot g(n)} \left|\sum_{a,b \,:\, a\oplus b = z} \left[\sum_i \sum_{(x,y) \in U_i} (-u_i) \mu(x,y) \Delta(a,b|x,y) +  \sum_j \sum_{(x,y) \in V_j} v_j \mu(x,y) \Delta(a,b|x,y)\right] \right|\\
                                        &\leq \frac{1}{2\cdot g(n)} \left[\sum_i \sum_{(x,y) \in U_i} |u_i| \mu(x,y) \left(\sum_{a,b} |\Delta(a,b|x,y)|\right)
+
\sum_j \sum_{(x,y) \in V_j} |v_j| \mu(x,y) \left(\sum_{a,b} |\Delta(a,b|x,y)|\right)\right] \\
                                        &\leq \frac{\epsilon}{2\cdot g(n)}  \left[\sum_i |u_i| \mu(U_i) + \sum_j |v_j| \mu(V_j)\right]
\end{align*}
From this calculation and \eqref{eq:corruption-BellInequality-Pf}, we obtain, for $\bp' = \bp_f + \bdelta$ :
\begin{align*}
B(\bp') = B(\bp_f) + B(\bdelta)
        \geq \frac{1}{2\cdot g(n)}\left [ \sum_{j}{v_j\mu(V_j)} - \epsilon \left(\sum_{j}{|v_j|\mu(V_j)} +\sum_{i}{|u_i|\mu(U_i)}\right) \right ],
\end{align*}
which proves \eqref{eq:corruption-BellInequality-Pfprime}.

\end{proof}

\section{Explicit examples}
Let us formulate a special case of Theorem~\ref{thm:boundDistribBellGeneral} that will be useful in the examples. Here there is just one subset in $f^{-1}(0)$ and one in $f^{-1}(1)$.
\begin{corollary}\label{cor:boundSpecialCase}
        Let $f$ be a (possibly partial) Boolean function on $\X\times\Y$, where $\X,\Y\subseteq\{0,1\}^n$. Given $\gamma \in (0,1)$ and $g: \ints \to (0,1)$, suppose that there is a distribution $\mu$ on $\X\times\Y$ such that: for any rectangle $R \subseteq \X\times\Y$,
\begin{equation}\label{eq:conditionMuSpecialCase}
        \mu(R\cap f^{-1}(1)) > \gamma \mu(R \cap f^{-1}(0)) - g(n).
\end{equation}
Then $\mu$ satisfies~\eqref{eq:conditionDistribBound} with $z = 0$, $i=j=1$, $U_1 = f^{-1}(1)$, $V_1 = f^{-1}(0)$, $u_1=1$, $v_1=\gamma$. Let $B$ be defined by~\eqref{eq:defGeneralB}, that is:
for all $a,b,x,y \in\ZO\times\ZO\times\X\times\Y$,
\begin{equation*}
B_{a,b,x,y}=\begin{cases}
-\frac{1}{2\cdot g(n)}\mu(x,y) & \mbox{ if } f(x,y)=1 \text{ and } a\oplus b=0 \\
\frac{\gamma}{2\cdot g(n)}\mu(x,y) & \mbox{ if } f(x,y)=0 \text{ and } a\oplus b=0 \\
0 & \mbox{otherwise.}
\end{cases}
\end{equation*}
Then, $B$ satisfies
\begin{align*}
B(\bl)&\leq 1, \quad \forall \bl \in \LL_{det}^{\bot},\\
        B(\bp_f)& = \frac{\gamma}{2\cdot g(n)}\mu(f^{-1}(0))
\end{align*}
and for any $\bp'\in\P$ such that $|\bp'-\bp_f|_1\leq \epsilon$ :
\begin{align*}
         B(\bp') \geq \frac{1}{2\cdot g(n)} \Big[ \gamma \mu(f^{-1}(0)) - \epsilon \big(\gamma \mu(f^{-1}(0)) + \mu(f^{-1}(1))\big) \Big].
\end{align*}

\end{corollary}
\subsection{Disjointness }\label{subsec:DISJ}


In~\cite{Razborov92}, Razborov proved the following.
\begin{lemma}[\cite{Razborov92}]\label{lem:lemmaRasborovDISJ}
There exist two distributions $\mu_0$ and $\mu_1$ with $\supp(\mu_0) \subseteq \DISJ_n^{-1}(1)$ and $\supp(\mu_1) \subseteq \DISJ_n^{-1}(0)$, such that: for any rectangle $R$ in the input space,
\begin{equation*}
	\mu_1(R) \geq \Omega(\mu_0(R)) - 2^{\Omega(n)}.
\end{equation*}
\end{lemma}
Following his proof, one can check that we actually have:
\begin{equation*}
	\mu_1(R) \geq \frac 1 {45} \mu_0(R) - 2^{-\epsilon n + \log_2(2/9)}.
\end{equation*}
So, letting $\mu := (\mu_0 + \mu_1)/2$,
\begin{equation}\label{eq:ineqDisjMu}
	\mu(R \cap f^{-1}(0)) \geq \frac 1 {45} \mu(R \cap f^{-1}(1)) - 2^{-\epsilon n + \log_2(4/9)}.
\end{equation}

\begin{remark}
	Actually, $\supp(\mu_1) = A_1 := \{(x,y) \,:\, |x|=|y|=m, |x \cap y|= 1\} \subseteq \DISJ_n^{-1}(0)$.
\end{remark}

Note that by this construction, $\mu(f^{-1}(0)) = \mu(f^{-1}(1)) = 1/2$. Combining \eqref{eq:ineqDisjMu} with Corollary~\ref{cor:boundSpecialCase} (with $g(n) = 2^{-\epsilon n + \log_2(4/9)}$), we obtain:
\begin{corollary}
	There exists a Bell inequality $B$ satisfying: $\forall \bl \in \LL_{det}^\bot,\>B(\bl) \leq 1$,
	\begin{equation*}
		B(\bp_{\DISJ_n}) = \frac 1{90} 2^{\epsilon n - \log_2(4/9)},
	\end{equation*}
	 and for any distribution $\bp'\in\P$ such that $|\bp'-\bp_{\DISJ_n}|_1\leq\eps$,
	\begin{equation*}
		B(\bp') \geq 2^{\epsilon n - \log_2(4/9)} \frac {1-46 \epsilon} {90}.
	\end{equation*}
\end{corollary}
More precisely, Theorem~\ref{thm:boundDistribBellGeneral} gives an explicit construction of such a Bell inequality: we can define $B$ as:
\begin{equation*}
B_{a,b,x,y}=\begin{cases}
-2^{\epsilon n - \log_2(4/9)} \mu(x,y) & \mbox{ if } \DISJ_n(x,y) =0 \text{ and } a\oplus b=1 \\
\frac{1}{45}2^{\epsilon n - \log_2(4/9)} \mu(x,y) & \mbox{ if } \DISJ_n(x,y)=1 \text{ and } a\oplus b=1 \\
0 & \mbox{otherwise.}
\end{cases}
\end{equation*}

To obtain Proposition~\ref{cor:DisjBellIneq}, we use Corollary~\ref{cor:existsBq} together with the fact that $Q_{\eps'}(\DISJ_n)=O(\sqrt{n})$.


\subsection{Tribes }\label{subsec:TRIBES}

Let $n=(2r+1)^2$ with $r\geq 2$ and let $\TRIBES_n:\{0,1\}^n\times\{0,1\}^n\to\{0,1\}$ be defined as:
\begin{equation*}
\TRIBES_n(x,y) := \bigwedge\limits_{i=1}^{\sqrt{n}} \left( \bigvee\limits_{j=1}^{\sqrt{n}} (x_{(i-1)\sqrt{n}+j}\text{ and } y_{(i-1)\sqrt{n}+j})\right).
\end{equation*}

In \cite{HJ13}[Sec. 3] the following is proven:
\begin{lemma}\label{lem:lemBoundTribes}
	There exists a probability distribution $\mu$ on $\{0,1\}^n\times\{0,1\}^n$ for which there exist numbers $\alpha,\lambda,\gamma,\delta>0$ such that for sufficiently large $n$ and for any rectangle $R$ in the input space:
$$\gamma\mu(U_1\cap R) \geq \alpha\mu(V_1\cap R) -\lambda\mu(V_2\cap R) -2^{-\delta n/2 + 1}$$
where $U_1 = \TRIBES_n^{-1}(0)$, $\{V_1,V_2\}$ forms a partition of $\TRIBES_n^{-1}(1)$ and $\mu(U_1)=1-7\beta^2/16$, $\mu(V_1)=6\beta^2/16$, $\mu(V_2)=\beta^2/16$ with $\beta=\frac{r+2}{r+1}$.
\end{lemma}
In \cite{HJ13}, the coefficients are $\alpha=0.99,\lambda=\frac{16}{3(0.99)^2}$ and $\gamma=\frac{16}{(0.99)^2}$ (the authors say these values have not been optimized).

Combining this result with our Theorem~\ref{thm:boundDistribBellGeneral} (taking $z=1, i=1,j=2,$ $U_1,V_1,V_2$ as in Lemma~\ref{lem:lemBoundTribes}, $u_1=\gamma, v_1=\alpha, v_2=-\lambda$, and $g(n) = 2^{-\delta n /2 + 1}$), we obtain:
\begingroup
\begin{corollary}
	There exists a Bell inequality satisfying: $\forall \bl \in \LL_{det}^\bot,\>B(\bl) \leq 1$,
	\begin{equation*}
		B(\bp_{\TRIBES_n}) = 2^{\delta n/2 - 1}\frac{\beta^2}{16}(6\alpha - \lambda),
	\end{equation*}
	and for any distribution $\bp'\in\P$ such that $|\bp'-\bp_{\TRIBES_n}|_1\leq\eps$,
	\begin{equation*}
		B(\bp') \geq 2^{\delta n/2 - 1}\left[\frac{\beta^2}{16}(6\alpha - \lambda)-\epsilon(\gamma(1-7\beta^2/16)+\lambda\beta^2/16+\alpha6\beta^2/16)\right].
	\end{equation*}
\end{corollary}
\endgroup
More precisely, Theorem~\ref{thm:boundDistribBellGeneral} provides a Bell inequality $B$ yielding this bound, defined as:
\begin{equation*}
B_{a,b,x,y}=\begin{cases}
-\gamma 2^{\delta n/2 - 1} \mu(x,y) & \text{if } (x,y)\in U_1 \text{ and } a\oplus b=1 \\
\alpha 2^{\delta n/2 - 1}  \mu(x,y) & \text{if } (x,y)\in V_1 \text{ and } a\oplus b=1 \\
-\lambda 2^{\delta n/2 - 1} \mu(x,y) & \text{if } (x,y)\in V_2 \text{ and } a\oplus b=1 \\
0 & \mbox{otherwise.}
\end{cases}
\end{equation*}

To obtain Proposition~\ref{coro:TribesBellIneq}, 
we use Corollary~\ref{cor:existsBq} together with the fact that $Q_{\eps'}(\TRIBES_n)=O(\sqrt{n}(\log n)^2)$.

\subsection{Gap Orthogonality}\label{subsec:ORT}

Let $f_n$ be the partial functions over $\{-1,+1\}^n\times\{-1,+1\}^n$ by $f_n(x,y) = \ORT_{64n}(x^{64}, y^{64})$, that is:
\begin{equation*}
f_n(x,y) =
	\begin{cases}
		-1 &\hbox{if } |\langle x,y\rangle| \leq \sqrt{n}/8 \\
		+1 &\hbox{if } |\langle x,y\rangle| \geq \sqrt{n}/4.
	\end{cases}
\end{equation*}

In~\cite{She12}, Sherstov proves the following result.
\begin{lemma}[\cite{She12}]\label{lem:GHDdistrib}
Let $\delta > 0$ be a sufficiently small constant and $\mu$ the uniform measure over $\{0,1\}^n\times\{0,1\}^n$ 
. Then, $\mu(f_n^{-1}(+1))=\Theta(1)$ and for all rectangle $R$ in $\{0,1\}^n\times\{0,1\}^n$ such that $\mu(R) > 2^{-\delta n}$,
\begin{equation*}
\mu(R\cap f_n^{-1}(+1))\geq \delta \mu(R\cap f_n^{-1}(-1)).
\end{equation*}

\end{lemma}

%
%

This implies that if we put uniform weight on inputs of $\ORT_{64n}$ of the form $(x^{64},y^{64})$ and put 0 weight on the others, we get a distribution $\mu'$ satisfying the constraints of Corollary \ref{cor:boundSpecialCase} for $\ORT_{64n}$ together with $\gamma=\delta$ from Lemma 4 and $g(64n)=2^{\delta n}$.

To get a distribution satisfying the constraints of Corollary \ref{cor:boundSpecialCase} on inputs of $\ORT_{64n+l}$ for all $0\leq l\leq 63$ we extend $\mu'$ as follows:
\begin{equation*}
\tilde\mu(xu,yv) =
	\begin{cases}
		\mu'(x,y) &\hbox{if } u=+1^l \text{, } v=-1^l\text{ and } \left(\langle x,y\rangle < -\sqrt{64n} \text{ or } 0\leq\langle  x,y\rangle \leq \sqrt{64n}\right)\\
		\mu'(x,y) &\hbox{if } u=+1^l\text{, } v=+1^l\text{ and } \left( -\sqrt{64n}\leq \langle x,y\rangle < 0 \text{ or } \langle  x,y\rangle > \sqrt{64n}\right)\\
        0 & \hbox{otherwise}
	\end{cases}
\end{equation*}

Using this distribution $\tilde\mu$ together with $\gamma=\delta$ from Lemma~\ref{lem:GHDdistrib} and with $g(n)=2^{-\delta n}$ we obtain, from Corollary~\ref{cor:boundSpecialCase}, a Bell inequality violation for $\ORT_{64n+l}$ for all $0\leq l \leq 63$:
\begin{corollary}
	There exists a Bell inequality $B$ satisfying: $\forall \bl \in \LL_{det}^\bot,\>B(\bl) \leq 1$,
	\begin{equation*}
		B(\bp_{\ORT_{64n+l}}) = 2^{\delta n} \delta \tilde\mu(\ORT_{64n+l}^{-1}(-1)),
	\end{equation*}
	and for any distribution $\bp'\in\P$ such that $|\bp'-\bp_{\ORT_{64n+l}}|_1\leq\eps$,
	\begin{equation*}
		B(\bp')\geq 2^{\delta n} \left(\delta \tilde\mu(\ORT_{64n+l}^{-1}(-1)) - \epsilon \big[ \delta \tilde\mu(\ORT_{64n+l}^{-1}(-1)) + \tilde\mu(\ORT_{64n+l}^{-1}(+1))\big]\right).
	\end{equation*}
\end{corollary}
More precisely, Theorem~\ref{thm:boundDistribBellGeneral} gives an explicit construction of such a Bell inequality: we can define $B$ as:
\[
B_{a,b,x,y}=\begin{cases}
- 2^{\delta n}\tilde\mu(x,y) & \mbox{if } (x,y)\in \ORT_{64n+l}^{-1}(+1) \text{ and } a\oplus b=-1 \\
\delta 2^{\delta n} \tilde\mu(x,y) & \mbox{if } (x,y)\in \ORT_{64n+l}^{-1}(-1) \text{ and } a\oplus b=-1 \\
0 & \mbox{otherwise.}
\end{cases}\]
%

To obtain Proposition~\ref{coro:ORTBellIneq},
we use Corollary~\ref{cor:existsBq} together with the fact that $Q_{\eps'}(\ORT_n)=O(\sqrt{n}\log n)$.

\section{Equivalent formulations of the efficiency bounds}
\label{app:equivEff}

In~\cite{LLR12}, the zero-error efficiency bound was defined in its primal and dual forms as follows

\begin{definition}[\cite{LLR12}]
\label{def:eff-zero}
The efficiency bound of a distribution $\bp\in\P$ is given by
\begin{align*}
 \eff(\bp)=&\min_{\zeta,\mu_\ell\geq 0}  &&\frac{1}{\zeta} \\
&\textrm{subject to } && \sum_{\ell\in\LL_{det}^\bot}\mu_\ell\ell(a,b|x,y)=\zeta p(a,b|x,y) & \forall (a,b,x,y)\in \A{\times}\B{\times}\X{\times}\Y\\
&&&\sum_{\ell\in\LL_{det}^\bot}\mu_\ell=1\\
=&\max_{B} && B(\bp)\\
&\textrm{subject to } && B(\bl)\leq 1 \quad \forall \bl\in\LL_{det}^\bot
\end{align*}
\end{definition}

The $\epsilon$-error efficiency bound was in turn defined as $\min_{\bp'\in\P|\bp'-\bp|_1\leq \epsilon} \eff(\bp')$. In this appendix, we show that this is equivalent to the definition used in the present article (Definition~\ref{def:eff}).
In the original definition, the Bell functional could depend on the particular
$\bp'$. We show that it is always possible to satisfy the constraint with
the same Bell functional for all $\bp'$ close to~$\bp$.

In order to prove this, we will need the following notions.
\begin{definition}
 A {\em distribution error}  $\bdelta$ is a family of additive error
terms $\Delta(a,b|x,y)\in [-1,1]$ for all $(a,b,x,y)\in \A{\times}\B{\times}\X{\times}\Y$ such that
\begin{align*}
 \sum_{a,b}\Delta(a,b|x,y)&=0 &\forall (x,y)\in\X\times\Y.
\end{align*}
For any $0\leq\epsilon\leq 1$, the set $\Delta_\epsilon$ is the set of distribution errors $\bdelta$ such that
\begin{align*}
 \sum_{a,b}|\Delta(a,b|x,y)|&\leq\epsilon &\forall (x,y)\in\X\times\Y.
\end{align*}
This set is a polytope, so it admits a finite set of extremal points. We denote this set by $\Delta_\epsilon^{ext}$.
\end{definition}

We will use the following properties of $\Delta_\eps$.
\begin{fact}\label{fact:error-1}
 For any distribution $\bp\in\P$, we have
\begin{align*}
 \{\bp'\in\P|\ |\bp'-\bp|_1\leq\epsilon\}\subseteq \{\bp+\bdelta|\ \bdelta\in\Delta_\eps\}
\end{align*}
\end{fact}
The reason why the set on the right-hand side might be larger is that $\bp+\bdelta$ might not be a valid distribution. In order to ensure that this is the case, it is sufficient to impose that all obtained purposed probabilities are nonnegative, leading to the following property.
\begin{fact}\label{fact:error-2}
  For any distribution $\bp\in\P$, we have
\begin{align*}
 \{\bp'\in\P|\ |\bp'-\bp|_1\leq\epsilon\}= \{\bp+\bdelta|\ \bdelta\in\Delta_\eps\ \&\ p(a,b|x,y)+\Delta(a,b|x,y)\geq 0\ \forall a,b,x,y\}
\end{align*}
\end{fact}

We are now ready to prove the following theorem.
\begin{theorem}
 Let $\bp\in\P$ be a distribution, $\eff_\eps(\bp)$ be defined as in Definition~\ref{def:eff} and $\eff(\bp)$ be defined as in Definition~\ref{def:eff-zero}. Then, we have
\begin{align*}
 \eff_\eps(\bp)=\min_{\bp'\in\P:|\bp'-\bp|_1\leq \epsilon} \eff(\bp').
\end{align*}
\end{theorem}

\begin{proof}
 Let $\overline{\eff}_\eps(\bp)=\min_{\bp'\in\P:|\bp'-\bp|_1\leq \epsilon} \eff(\bp')$. We first show that $\eff_\eps(\bp)\leq\overline{\eff}_\eps(\bp)$. Let $(B,\beta)$ be an optimal feasible point for $\eff_\eps(\bp)$, so that
\begin{align*}
 \eff_\eps(\bp)&=\beta,\\
B(\bp')&\geq\beta&\forall&\bp'\>\rm{\>s.t.\>}|\bp'-\bp|_1\leq\epsilon,\\
B(\bl)&\leq 1 & \forall& \bl\in\LL_{det}^\bot.
\end{align*}
Therefore $(B,\beta)$ is also a feasible point for $\eff(\bp')$ for all $\bp'\in\P$ such that $|\bp'-\bp|_1\leq\epsilon$, so that $\eff(\bp')\geq\beta$ for all such $\bp'$, and $\overline{\eff}_\eps(\bp)\geq\beta=\eff_\eps(\bp)$.

It remains to show that $\eff_\eps(\bp)\geq\overline{\eff}_\eps(\bp)$. In order to do so, we first use the primal form of $\eff(\bp')$ in Definition~\ref{def:eff-zero} to express $\overline{\eff}_\eps(\bp)$ as follows
\begin{align*}
 \overline{\eff}_\eps(\bp)
=&\min_{\substack{\bp'\in\P\\\textrm{s.t }|\bp'-\bp|_1\leq \epsilon}} &&\eff(\bp')\\
=&\min_{\zeta,\mu_\ell\geq 0,\bp'\in\P}  &&\frac{1}{\zeta} \\
&\textrm{subject to } && \sum_{\ell\in\LL_{det}^\bot}\mu_\ell\ell(a,b|x,y)=\zeta p'(a,b|x,y) & \forall (a,b,x,y)\in \A{\times}\B{\times}\X{\times}\Y\\
&&&\sum_{\ell\in\LL_{det}^\bot}\mu_\ell=1,\quad |\bp'-\bp|_1\leq\epsilon\\
=&\min_{\zeta,\mu_\ell\geq 0,\bdelta\in\Delta_\eps}  &&\frac{1}{\zeta} \\
&\textrm{subject to } && \sum_{\ell\in\LL_{det}^\bot}\mu_\ell\ell(a,b|x,y)= \\&&&\quad\quad \zeta [p(a,b|x,y)+\Delta(a,b|x,y)] & \forall (a,b,x,y)\in \A{\times}\B{\times}\X{\times}\Y\\
&&&\sum_{\ell\in\LL_{det}^\bot}\mu_\ell=1,
\end{align*}
where the last equality follows from Fact~\ref{fact:error-2} and the fact that the first condition of the program imposes that $p(a,b|x,y)+\Delta(a,b|x,y)$ is nonnegative (since $\sum_\ell\mu_\ell\ell(a,b|x,y)$ is nonnegative). Since $\Delta_\eps$ is a polytope, $\overline{\eff}_\eps(\bp)$ can be expressed as the following linear program
\begin{align*}
 \overline{\eff}_\eps(\bp)
=&\min_{\zeta,\mu_\ell\geq 0,\nu_\Delta\geq 0}  &&\frac{1}{\zeta} \\
&\textrm{subject to } && \sum_{\ell\in\LL_{det}^\bot}\mu_\ell\ell(a,b|x,y)=\zeta [p(a,b|x,y)+ \\
&&&\quad \sum_{\bdelta\in\Delta_\eps^{ext}}\nu_\Delta\Delta(a,b|x,y)] & \forall (a,b,x,y)\in \A{\times}\B{\times}\X{\times}\Y\\
&&&\sum_{\ell\in\LL_{det}^\bot}\mu_\ell=1,\quad\sum_{\bdelta\in\Delta_\eps^{ext}}\nu_\Delta=1.
\end{align*}
Note that this can be written in standard LP form via the change of variables $\zeta' = 1/\zeta, \mu_\ell = \zeta w_\ell $. By LP duality, we then obtain
\begin{align*}
  \overline{\eff}_\eps(\bp)
=&\max_{B,\beta} && \beta\\
&\textrm{subject to } &&
B(\bp+\bdelta)\geq \beta & \forall \bdelta&\in\Delta_\eps, \\
&&&B(\bl)\leq 1 & \forall \bl&\in\LL_{det}^\bot.
\end{align*}
Comparing this to the definition of $\eff_\eps(\bp)$ (Definition~\ref{def:eff}) and together with Fact~\ref{fact:error-1}, we therefore have $\overline{\eff}_\eps(\bp)\leq\eff_\eps(\bp)$.

\end{proof}

\end{document}